\documentclass[12pt]{article}

\usepackage{url}
\usepackage{bbm}
\usepackage{mathtools}
\usepackage{natbib}
\usepackage{amssymb}
\usepackage{amsthm}
\usepackage{empheq}
\usepackage{latexsym}
\usepackage{enumitem}
\usepackage{eurosym}
\usepackage{dsfont}
\usepackage{appendix}
\usepackage{algorithm}
\usepackage{algorithmicx}
\usepackage{algcompatible}
\usepackage{color}
\usepackage{subcaption}
\usepackage[unicode]{hyperref}
\usepackage{frcursive}
\usepackage[utf8]{inputenc}
\usepackage[T1]{fontenc}
\usepackage{geometry} 
\usepackage{multirow}
\usepackage{todonotes}
\usepackage{lmodern}
\usepackage{anyfontsize}
\usepackage{pgfplots}
\usepackage{stmaryrd}
\usepackage{float}
\usepackage{bookmark}

\graphicspath{ {images/} }

\pgfplotsset{compat=newest}

\definecolor{red}{rgb}{0.7,0.15,0.15}
\definecolor{green}{rgb}{0,0.5,0}
\definecolor{blue}{rgb}{0,0,0.7}
\hypersetup{colorlinks, linkcolor={red},citecolor={green}, urlcolor={blue}}
            
\makeatletter \@addtoreset{equation}{section}

\newtheorem{theorem}{Theorem}
\newtheorem{theorem2}{Theorem}[section]
\newtheorem{assumption}[theorem2]{Assumption}
\newtheorem{corollary}[theorem2]{Corollary}

\newtheorem{lemma}[theorem2]{Lemma}
\newtheorem{proposition}[theorem2]{Proposition}

\newtheorem{definition}[theorem2]{Definition}
\newtheorem{remark}[theorem2]{Remark}

\DeclareUnicodeCharacter{014D}{\=o}

\def \E{\mathbb{E}}

\def \H{\mathbb{H}}

\def \K{\mathbb{K}}
\def \L{\mathbb{L}}

\def \P{\mathbb{P}}

\def \R{\mathbb{R}}
\def \S{\mathbb{S}}

\title{Delegation or competition: optimal liquidation across dark and lit pools with information privilege}
\author{Thibaut \textsc{MASTROLIA} and Hao \textsc{WANG}\thanks{\texttt{mastrolia@berkeley.edu}, \texttt{haowang013@berkeley.edu}.}\\ UC Berkeley\\
Department of Industrial Engineering and Operations Research}

\date{}

\begin{document}
\maketitle
\begin{abstract}
We study the optimal liquidation problem in both dark and lit pools for an investor delegating the execution of a large position to a broker, in a continuous-time setting with market impact and execution uncertainty. First, we design an optimal compensation scheme between the investor and a broker who owns dark pools, may internalize the client's order flow, and enjoys an information privilege over the rest of the market. We explicitly characterize the broker's optimal internalization, hedging, and routing strategies across the lit venue and the dark pools using BSDE theory. Second, we study a purely competitive model in which the investor liquidates directly on the open market against a crowd of strategic high-frequency traders, leading to a major-minor mean-field game for which we provide explicit solutions to the associated HJB Fokker-Planck system. Finally, we illustrate our results through numerical experiments, comparing the welfare and the price impact of the delegation regime with those of direct execution, and quantifying the economic value of broker intermediation.
\end{abstract}


\section{Introduction}

\subsection{Optimal execution and price impact}
Optimal execution is the process of executing large orders in financial markets to minimize trading costs and market impact induced by large trades. It is a cornerstone of modern financial research. The seminal work of \cite{almgren2001optimal} introduced a mean-variance framework that optimizes trading trajectories by considering both execution speed and price impact, assuming a linear impact function. This framework has been extended to incorporate more complex market dynamics. For instance, \cite{bertsimas1998optimal} explored optimal control methods for execution costs, introducing models that account for nonlinear price impacts. \cite{huberman2004optimal} focused on liquidity trading, emphasizing the interplay between liquidity provision and price impact. \cite{bouchaud2009how} analyzed the slow digestion of supply and demand shocks in markets, revealing temporal patterns in price adjustments that persist beyond immediate trades. \cite{obizhaevawang2013optimal} introduced a supply-demand dynamic model, illustrating how order book shapes influence optimal trading strategies. More recently, \cite{forde2022optimal} analyzed execution strategies in markets with Gaussian signals and power-law recovery, offering insights into how market resilience affects trading decisions. \cite{cont2017optimal} investigated optimal order placement in limit order books, demonstrating how strategic placement across price levels can enhance execution efficiency. Continuous-time models have been developed in, for example, \cite{schied2009risk,gueant2012optimal}. Extensions to optimal liquidation problems with a large number of traders in mean-field games or mean-field control problems have been initiated in \cite{cardaliaguet2018mean} then investigated in the series of articles \cite{fu2021mean,fu2024mean1,fu2024mean2}, while \cite{baldacci2023mean, bergault2025mean,euch2021optimal,barucci2025market} focus on trading strategies on large scales, information asymmetry and games between market participants. Reinforcement learning methods for optimal liquidation have also been developed in, for example, \cite{hendricks2014reinforcement} and recently in \cite{graf2026learning}, anticipating closing auction mechanisms. These studies collectively highlight the evolving complexity of optimal execution models, incorporating factors such as stochastic volatility, trade-off between fast liquidation but high market impact, and liquidity constraints.

\subsection{Dark pool trading}
Optimal liquidation has been extensively investigated in the literature in the context of a classical venue, the lit pool, as emphasized above.
The complexity of this task has grown with the advent of alternative trading venues and the fragmentation of financial markets, such as dark pools, which allow anonymous trading but introduce challenges to price discovery and market transparency. The theoretical analysis of dark venues goes back to \cite{hendershott2000crossing}, who model the competition between a crossing network and a dealer market. Empirical studies, such as \cite{degryse2015impact}, found that dark pools can enhance liquidity for large trades but may fragment markets, potentially widening spreads. \cite{zhu2014dark} investigated whether dark pools impair price discovery, finding mixed effects depending on market conditions. Similarly, \cite{comertonforde2015dark} noted that while dark pools can improve liquidity for institutional traders, they may increase volatility and spreads in certain scenarios, highlighting the nuanced impact of dark pools on market quality. Although in general dark pools have become integral to institutional trading by reducing market impact for large orders, the liquidity is uncertain for investors. \cite{ganchev2010censored} modeled dark pool allocation as a multiarmed bandit problem, using adaptive algorithms to distribute orders across multiple dark pools. \cite{kratz2014optimal, kratz2018optimal} examined optimal liquidation in dark pools, analyzing the trade-offs between execution speed and adverse selection risks inherent in dark pool trading. \cite{perlman2021dark} approached dark pool trading from a stochastic control perspective, proposing adaptive allocation strategies to optimize execution across lit and dark venues. Recently, \cite{baldacci2023market} investigates the design of incentive schemes to improve market allocation between lit and dark pools with latency and market impact adapting a deep learning actor-critic method to solve the bilevel optimization considered.\\

\subsection{Broker intermediation and information privilege}

While \cite{baldacci2023market} proposes a model in which an exchange incentivizes a market maker, most of the practical studies consider the delegation of a portfolio to brokers owning dark pools and rerouting orders to either dark or lit venues, see \cite{battalio2016can}. This implies order-flow internalization, when the broker, instead of routing the client's delegated order to an external venue, fills it inside its own dark pool capturing the spread and tape revenue rather than seeking best execution elsewhere, often against its own proprietary flow or other clients leading to adverse selection effects. See for example \cite{easley1996cream,chakravarty2002model}. The articles \cite{battalio1997competing,battalio2024cost} investigate in particular dark pool operators
and the possibility to waive execution fees for brokers in exchange for access to order flow information, an arrangement that effectively benefits both the broker and the dark pool operator at the expense of the large order trader.\\

A key feature of this intermediation is the information privilege of the broker: the broker observes the client's entire liquidation schedule in real time, whereas the rest of the market only sees the displayed lit-market flow, and the exploitation of such client-order information by brokers is empirically documented in \cite{barbon2019brokers}. Internalized and dark-pool fills typically occur at the mid-price with a small price improvement, in line with \cite{weaver2011internalization,chung2006does,battalio2003all}. This informational asymmetry gives the broker a timing advantage: by delaying the exposure of the client's order in its own dark pools until the trading footprint has largely materialized, the broker reduces adverse selection and may harvest the mean reversion of the transient price impact, see \cite{grammig2012best}. The attractiveness of dark venues is further tied to the relative bid-ask spread of the underlying stock: \cite{buti2022diving} document that dark-pool order routing increases when the relative tick size is large, as traders seek to avoid the cost of crossing wide spreads or to bypass long queues at the top of the limit order book.\\

This problem can thus be assimilated to a principal-agent problem of portfolio delegation type in which an investor (client or the principal) delegates their own portfolio to an agent (the broker). While the investor liquidates directly on a lit pool generating both permanent and transient impact, the broker may buy shares directly from the client and liquidate them on their dark pools (internalization), together with their own liquidation strategy on the lit pool. The broker thus receives a compensation from the client to access these dark pool venues, removing transaction fees as suggested in \cite{battalio2024cost}. On top of that, both of them are faced with high-frequency traders operating on the lit market and trading around large institutional orders, see \cite{van2019high}. This is the new contracting problem we are considering here, taking into account market impact on the dark pools, internalization and uncertainty of execution in the dark pools that we describe more in the section below.

\subsection{Problem investigated and main contributions}

\subsubsection{Contract theory for optimal liquidation, market impact and execution uncertainty}
Contract theory is a branch of economics that studies the design of agreements between parties with potentially conflicting interests, particularly under conditions of asymmetric information. A central concern is the asymmetry of decision or information, where the principal cannot directly index the contract in the agent's effort to manage a risky project. The theory provides frameworks for designing contracts that align incentives. Formulated as a principal-agent problem, contract theory uses tools from game theory and mechanism design to determine optimal contractual arrangements. The fundamental work of H\"olmstrom and Milgrom \cite{holmstrom1987aggregation} introduced a first continuous-time framework with Brownian motion. This theory has regained interest, especially in the mathematical community in 2000-2010 \cite{sannikov2008continuous,cvitanic2012contract,cvitanic2018dynamic,hubert2023continuous,chiusolo2024new,hernandez2024closed} and applications in optimal asset allocations and delegated portfolio management \cite{cvitanic2017moral,gennaro2024delegated,xing2025optimal}, delegated order execution \cite{larsson2025optimal}, market design \cite{baldacci2024design}, market making \cite{euch2021optimal,baldacci2021optimal,barzykin2023algorithmic} or auctions trading \cite{mastrolia2024clearing,mastrolia2025optimal}.\\

In terms of market quality, price efficiency (or price impact) is a key metric to evaluate trading venues. The rise of dark pools and market fragmentation has led to extensive research on their effects. \cite{ohara2011market} found that fragmentation can improve liquidity and reduce trading costs, challenging the notion that it is inherently detrimental. Conversely, \cite{chen2021market} highlighted potential downsides, such as reduced transparency and increased routing complexity. Regulatory frameworks, such as the Markets in Financial Instruments Directive (MiFID), aim to enhance market quality through transparency and competition. \cite{cimon_2017} analyzed MiFID's impact, finding improvements in liquidity but mixed effects on volatility. The interplay between market structure and quality remains a critical area of study, particularly as alternative venues proliferate.\\

This study contributes to this literature by proposing the continuous-time principal-agent problem for delegated liquidation in which the agent is a broker owning dark pools, with uncertainty of execution and internalization of the client's order flow. We investigate an investor offering a compensation scheme to a broker liquidating a large position across one lit venue and several dark pools, the broker enjoying an information privilege over the rest of the market.

\subsubsection{Major-minor traders and optimal liquidation}

Financial markets are competitive and consist of many types of traders, especially high-frequency traders who target large trades. Mean-field games (MFGs) provide a robust framework for modeling strategic interactions among numerous agents, making them suitable for analyzing competitive trading environments.
The theory of mean field games has been developed in \cite{huang2006large,lasry2007mean} with an analytical approach, see for example \cite{cardaliaguet2010notes} for a review of the method used with PDE, and in \cite{carmona2018probabilistic,carmona2015probabilistic} for a probabilistic approach of the equilibrium characterization. Regarding applications in finance,
\cite{fujii_takahashi_2020} applied MFGs to equilibrium pricing under market clearing conditions, offering a framework for analyzing market stability. \cite{huang2019mean} and \cite{cardaliaguet2018mean} explored MFGs in trading congestion scenarios, showing how collective trader behavior amplifies market impact and costs. To investigate the interaction between a large trader and small traders, \cite{carmona_lacker_2015} provided a probabilistic MFG approach for major-minor players. \cite{firoozi2017execution} formulated execution problems with major and minor traders using MFGs, aligning closely with our comparison of competitive markets. \cite{fu2021mean} analyzed optimal portfolio liquidation in an MFG framework, deriving equilibrium strategies and their impact on prices. \cite{chen2024periodic} studied periodic trading patterns in financial markets using MFGs with major and minor players, providing insights into how  competitive dynamics shape trading behavior. These works highlight the applicability of MFGs to modeling complex, multi-agent trading environments.\\

As an extension of our primary broker-delegation study, we incorporate strategic high frequency traders and contribute to the major-minor players literature by proposing a new model for optimal liquidation under uncertainty between a strategic large trader faced with a crowd of strategic small traders, leading to a new system of Fokker-Planck equations that we solve explicitly with closed form solutions.\\

\subsubsection{Deep learning methods for PDEs}

All the problems presented above are reduced to systems of high dimensional PDE, with embedded problems for the principal-agent method. To solve the involved high-dimensional PDEs, \cite{raissi2024forward} and \cite{ji2022deep} developed powerful deep learning methods. To make such a framework applicable to financial market, \cite{baldacci2023market} used deep reinforcement learning to design incentives and optimize trading in dark pools. \cite{lu2025multiagent} introduced a computational framework based on the actor-critic method in deep reinforcement learning to solve the stochastic control problem with jumps.\\

We adapt these methods to our problem, by considering the uncertainty of execution in the broker's dark pools and embedding the broker's best response into the learning procedure, while the competitive benchmark admits explicit solutions and requires no learning.\\

The paper is organized as follows. Section~\ref{sec:mechanisms} presents the market mechanisms and the main actors. Section~\ref{sec:market} sets up the weak formulation of the model, solves the broker's optimization problem under linear and exponential utility, and reduces the investor's contracting problem to a Markovian control problem with a verification theorem. Section~\ref{sec:mfg_hft} studies the direct-execution benchmark as a major-minor mean-field game between the investor and strategic high-frequency traders. Section~\ref{sec:numerics} presents the numerical methods and the Monte Carlo experiments comparing the two regimes, and Section~\ref{sec:conclusion} concludes.

\section{Market Mechanisms overview}\label{sec:mechanisms}

We first present the main actors of the market and the general big picture of the broker-delegation faced with high frequency traders. We set the main notations and strategies for each participant informally; rigorous definitions are given in the next section.  We are considering:
\begin{itemize}[leftmargin=1.2em]
  \item \textbf{An Investor} holding a large initial position $Q_0>0$ at time $0$, who wishes to liquidate it completely over $[0,T]$. To avoid paying transaction fees and crossing the spread, the investor delegates execution to a broker, and the position is liquidated on both the common lit pool and the dark pools the broker owns. The investor proposes a compensation $\xi$ as an $\mathcal F_T$-measurable random variable to the broker for their trading activities;

  \item \textbf{a crowd of high-frequency traders} with initial inventory $0$. They can only observe the trading rate in the lit market and send market orders at rate $\varsigma$;

  \item \textbf{a broker} who provides execution services and operates $M$ dark pools of
        its own. The broker may use the following three ways to deal with the investor's position: (i) place part of the order in the lit
        market on behalf of the investor, with trading strategy $\nu$; (ii) route it to its dark pools with volume $\ell$; or (iii) internalize
        the flow, with strategy $u$, becoming the counterparty itself. In addition, the broker has their own hedging strategy $h$ on the lit market directly.
\end{itemize}

To summarize the venues accessible to these participants, we have
\begin{itemize}[leftmargin=1.2em]
  \item \textbf{a common lit market,} in which 
        market orders are executed and generate both permanent and transient price impact. Market orders cross a constant half–spread $\delta>0$
        and incur instantaneous impact $\lambda>0$. We assume that permanent and transient price impact follow the
        OU–propagator of \cite{obizhaevawang2013optimal}.

  \item \textbf{$M$ broker-owned dark pools} which match the volume submitted internally by the broker, intermittently executed at prices tied to the contemporaneous mid-price with a small mandatory improvement
        (see \citealp{zhu2014dark}), that is to say all fills (both the internalized trades and matching trades) occur at mid-price $\pm$ a tiny improvement $\varepsilon^{(i)}$ with $|\varepsilon^{(i)}|\ll \delta$,
        in line with \cite{weaver2011internalization, chung2006does, battalio2003all}. Matching in pool $i$ arrives as a Poisson process with
        state-dependent intensity
        $\lambda_t^{(i)}=\lambda_0^{(i)}(\delta)$.
\end{itemize}

Finally we point out that the broker observes the entire liquidation schedule in real time, as an information privilege
whereas the rest of the market only sees the lit venue.  
This asymmetry motivates the broker to
delay dark-pool exposure, reduce adverse selection, and harvest
transient-impact mean reversion once the investor’s trading
footprint has largely materialized \cite{grammig2012best}.

\section{Mathematical Model}\label{sec:market}

\subsection{Market model}
\label{sec:market-model}

Before introducing the model, we motivate the choice of a weak formulation,
following the problem initiated by Holmström and Milgrom \cite{holmstrom1987aggregation} then in
\cite{cvitanic2018dynamic}. The investor observes only the price path $S$ and the realized
fills; they cannot observe the broker's internal decomposition of that path
into $\nu$ (internalization, execution on behalf of the client), $h$ (hedging of the broker), $\ell$ (dark pool strategy) and the
crowd's $\varsigma$ (HFT activity).  The weak formulation fixes the underlying sample
space and exogenous randomness once and for all, and lets each admissible
control correspond instead to a change of the probability measure on
that same space, via Girsanov's theorem. Under every such measure $S$ retains
identical pathwise regularity and quadratic variation; only its law
changes. This is the operational meaning of the broker controlling the law of
the process rather than the process itself, and it is what makes the
principal's contracting problem in Section~\ref{sec:broker} well posed as a single optimization
over $(Y_0,Z,K)$: every contract $\xi$ must be measurable with
respect to what is publicly observed, and the broker's hidden action can only
ever manifest through the law those observables follow.\\

Fix a filtered probability space $(\Omega,\mathcal F,(\mathcal F_t)_{0\le t\le T},\mathbb P^0)$,
$T>0$ the liquidation horizon endowed with:
\begin{itemize}
\item a standard $\mathbb P^0$-Brownian motion $W$;
\item for $i=1,\dots,M$, independent $\mathbb P^0$-Poisson processes
$N^{(i)}$ with intensity $\lambda_0^{(i)}(\delta)>0$ (possibly depending on the
prevailing half-spread $\delta$), mutually independent and independent of $W$,
recording the matching events in dark pool $i$;
\item for each $i$, a random variable $r^{(i)}$ with CDF $F_i(r)$,
independent of $W$ and of $(N^{(j)})_{j\ne i}$, representing the
exogenous liquidity available for matching in pool $i$ at each jump of
$N^{(i)}$;
\item an exogenous predictable benchmark process
$G^{(i)}=(G^{(i)}_t)_{0\le t\le T}$, representing the volume the broker
would route to pool $i$ under a passive reference policy as for example a fixed
participation schedule, independent of any control choice.
\end{itemize}
$(W,(N^{(i)},r^{(i)},G^{(i)})_{i\le M})$ are the benchmark characteristics of the
model: they are fixed once and for all on $\Omega$, and every control
introduced below acts only through a change of the probability measure under
which they are observed.\\ 

Throughout, all spaces below are understood relative to the filtered
probability space $(\Omega,\mathcal F,(\mathcal F_t)_{0\le t\le T},\mathbb P)$
in force.

\begin{itemize}
\item $\mathbb S^2_T(\mathbb R)$ denotes the space of
$(\mathcal F_t)$-adapted, càdlàg processes $Y=(Y_t)_{0\le t\le T}$ with
values in $\mathbb R$ such that
\[
\|Y\|^2_{\mathbb S^2_T}:=\mathbb E^{\mathbb P}\left[\sup_{0\le t\le T}|Y_t|^2\right]<\infty.
\]

\item $\mathbb H^2_T(\mathbb R)$ denotes the space of
$(\mathcal F_t)$-predictable processes $Z=(Z_t)_{0\le t\le T}$ with values
in $\mathbb R$ such that
\[
\|Z\|^2_{\mathbb H^2_T}:=\mathbb E^{\mathbb P}\left[\int_0^T|Z_t|^2\,dt\right]<\infty.
\]

\item $\mathbb K^2_T(\mathbb R^M)$ denotes the space of
$(\mathcal F_t)$-predictable processes $K=(K^{(1)}_t,\dots,K^{(M)}_t)_{0\le t\le T}$
with values in $\mathbb R^M$ such that
\[
\|K\|^2_{\mathbb K^2_T}:=\mathbb E^{\mathbb P}\left[\sum_{i=1}^M\int_0^T|K^{(i)}_t|^2\,\lambda_0^{(i)}(\delta)\,dt\right]<\infty,
\]
equivalently, $\|K\|^2_{\mathbb K^2_T}=\mathbb E^{\mathbb P}\big[\sum_{i=1}^M\langle K^{(i)}\!\cdot\widetilde N^{(i)}\rangle_T\big]$,
the expected predictable quadratic variation of the stochastic integrals
$\int_0^\cdot K^{(i)}_s\,d\widetilde N^{(i)}_s$, $i=1,\dots,M$.
\end{itemize}
Each of $(\mathbb S^2_T(\mathbb R),\|\cdot\|_{\mathbb S^2_T})$,
$(\mathbb H^2_T(\mathbb R),\|\cdot\|_{\mathbb H^2_T})$, and
$(\mathbb K^2_T(\mathbb R^M),\|\cdot\|_{\mathbb K^2_T})$ is a Hilbert space
(Banach in the case of $\mathbb S^2_T$), and we equip the product space
$\mathbb S^2_T(\mathbb R)\times\mathbb H^2_T(\mathbb R)\times\mathbb K^2_T(\mathbb R^M)$
with the norm
\[
\|(Y,Z,K)\|:=\Big(\|Y\|^2_{\mathbb S^2_T}+\|Z\|^2_{\mathbb H^2_T}+\|K\|^2_{\mathbb K^2_T}\Big)^{1/2}.
\]

We first define the dark-pool fill process associated with the
benchmark $G$, with no
control and no probability change involved yet:
\[
dJ^{(i),0}_t := \min\big\{G^{(i)}_t,\,r^{(i)}\big\}\, dN^{(i)}_t,
\qquad J^{(i),0}_0=0,\quad i=1,\dots,M.
\]
$J^{(i),0}$ records how much of the benchmark volume $G^{(i)}$ is actually
matched in pool $i$: at each arrival of $N^{(i)}$ the fill is capped by
whatever counterparty liquidity $r^{(i)}$ happens to be available at that
instant.

The broker's actual routing decision is to replace the passive benchmark
$G^{(i)}$ by a chosen predictable process $\ell^{(i)}$, jointly with a lit-market
rate $\nu$, a hedging rate $h$, an internalization rate $u$, and an HFT rate
$\varsigma$. We now define the set of admissible controls, and the change of
probability measure through which $(\nu,h,u,\varsigma, \ell)$ act on the model.\\

The routing control $\ell$ requires no change of measure of its own since
$N^{(i)}$ and $r^{(i)}$ are independent of $W$ under $\mathbb P^0$, and the
Girsanov step below acts only on $W$, so it leaves the joint law of
$(N^{(i)},r^{(i)})_{i\le M}$ untouched. The passage from the benchmark fill
$J^{(i),0}$ to the broker's controlled fill is thus a purely pathwise
substitution of $\ell^{(i)}$ for $G^{(i)}$, valid
under $\mathbb P^0$ and under every measure $\mathbb P^\alpha$ defined rigorously below, so that:
\[
J^{(i),\alpha}_t:=\int_0^t\min\{\ell^{(i)}_s,r^{(i)}\}\,dN^{(i)}_s,
\qquad
J^{(i),\alpha}_0=0,\qquad i=1,\dots,M.
\]

Let $\sigma>0$ denote mid-price volatility, $\gamma,\varepsilon>0$ the permanent price-impact coefficients of the broker's and the HFT's flow respectively. Trading in the lit market generates both permanent and transient price
impact. Following \cite{obizhaevawang2013optimal}, we model these effects using
a linear propagator with exponential decay. The transient impact $I_t$ follows
\[
  dI_t = -\kappa I_t\,dt + \eta\bigl(\nu_t + h_t + \varsigma_t \bigr)\,dt,
  \qquad I_0 = 0.
\]

\begin{definition}[Admissible controls]
An admissible control is a tuple $\alpha=(\nu,h,u,\varsigma,\ell)$ of
$(\mathcal F_t)$-predictable processes such that $(\nu_t,h_t)\in\mathbb R^2$,
$u_t\in[0,\infty)$, $\varsigma_t\in\mathbb R$, and
$\ell_t=(\ell^{(i)}_t)_{i\le M}$ satisfies $\ell^{(i)}_t\ge0$ and
$\sum_{i=1}^M\ell^{(i)}_t\le Q^C_t$, where
\[
Q^C_t:=Q_0-\int_0^t(\nu_s-u_s)\,ds-\sum_{i=1}^M J^{(i),\alpha}_t,
\]
together with the integrability condition
\[
\mathbb E^{\mathbb P^0}\!\left[\int_0^T\Big(\nu_s^2+u_s^2+h_s^2+\varsigma_s^2+\|\ell_s\|^2\Big)ds\right]<\infty,
\]
and such that
\[
\theta_s^{\nu,h,\varsigma}:=\frac{(\gamma+\eta)(\nu_s+h_s)+(\epsilon+\eta)\varsigma_s}{\sigma}
\]
generates a $\mathrm{BMO}(\mathbb P^0)$ martingale $\theta^{\nu,h,\varsigma}\!\cdot W$.
We denote by $\mathcal A$ the set of admissible controls.
\end{definition}
We refer to \cite{kazamaki2006continuous} for more details on BMO martingale theory especially important for the exponential utility problem in Section \ref{sec:expo}. 
Let
$Q_t^{B}$ be the broker's inventory arising from internalization, and $Q_t^H$ the high frequency trader's inventory.
Their dynamics are given by
\[\text{(Inventories)}\begin{cases}
\text{(Investor)} & dQ_t^{C}
  = (\nu_t - u_t)dt
    -\sum_{i=1}^{M} dJ_t^{(i),\alpha}, \qquad Q_0^{C}=Q_0,\\
\text{(Broker)} & dQ_t^{B}
  = (u_t + h_t)\,dt, \qquad Q_0^{B}=0,\\
\text{(HF Traders)} & dQ_t^{H}
  = \varsigma_t\,dt, \qquad Q_0^{H}=0,
  \end{cases}
\]
where we recall that $u_t$ is the internalization rate and $J_t^{(i),\alpha}$
denotes cumulative executions in dark pool $i$.

\begin{definition}[Controlled measure]
For $\alpha=(\nu,h,u,\varsigma,\ell)\in\mathcal A$, let
\[
L^{\nu,h,\varsigma}_t:=\mathcal E\Big(\int_0^\cdot\theta_s^{\nu,h,\varsigma}\,dW_s\Big)_t
=\exp\left(\int_0^t\theta_s^{\nu,h,\varsigma}\,dW_s-\frac12\int_0^t\big(\theta_s^{\nu,h,\varsigma}\big)^2ds\right).
\]
By the $\mathrm{BMO}$ condition and Kazamaki's criterion, $L^{\nu,h,\varsigma}$
is a uniformly integrable $(\mathcal F_t,\mathbb P^0)$-martingale, and
$\frac{d\mathbb P^\alpha}{d\mathbb P^0}\bigg|_{\mathcal F_t}:=L^{\nu,h,\varsigma}_t
$
defines a probability measure $\mathbb P^\alpha$ equivalent to $\mathbb P^0$ on $\mathcal F_T$.
\end{definition}

\begin{proposition}
Let $\alpha=(\nu,h,u,\varsigma,\ell)\in\mathcal A$ and $\mathbb P^\alpha$ as
above. Set $W^\alpha_t:=W_t-\int_0^t\theta_s^{\nu,h,\varsigma}\,ds$. Then:
\begin{enumerate}
\item[(i)] $W^\alpha$ is a standard $(\mathcal F_t,\mathbb P^\alpha)$-Brownian
motion, and the mid-price process with impact under $\mathbb P^\alpha$ has the following dynamic,
\[
dS_t=\sigma\,dW^\alpha_t+\Big[(\gamma+\eta)(\nu_t+h_t)+(\epsilon+\eta)\varsigma_t-\kappa I_t\Big]dt,\; S_0>0.
\]
\item[(ii)] Under $\mathbb P^\alpha$, $N^{(i)}$ remains a Poisson process of
intensity $\lambda_0^{(i)}(\delta)$ and $r^{(i)}$ retains its $\mathbb P^0$-law,
for every $i=1,\dots,M$. Consequently $J^{(i),\alpha}_t=\int_0^t\min\{\ell^{(i)}_s,r^{(i)}\}\,dN^{(i)}_s$
holds identically under $\mathbb P^\alpha$, with predictable compensator $\min\{\ell^{(i)}_t,r^{(i)}\}\,\lambda_0^{(i)}(\delta)\,dt.$
\end{enumerate}
\end{proposition}
For more details on this proof, we refer for example to \cite[Lemma 2.6]{mastrolia2025optimal}.\\

We now introduce the notion of internalized prices. A market order of size $c_t$ executed on the
lit venue is filled at
\[
P^{\mathrm{lit}}_t=S_t+\lambda c_t+\delta\,\mathrm{sgn}(c_t),
\]
with $\delta>0$ the constant half bid–ask spread and $\lambda>0$ the
instantaneous impact coefficient. In this paper, $c_t$ may denote the investor's strategy $\nu_t$, the broker's hedging rate $h_t$, or the HFTs' trading rate $\varsigma_t$. Dark-pool and internalized executions are
priced at a small, fixed improvement to the contemporaneous mid-price,
\[
P^{(i)}_t=S_t+\varepsilon^{(i)},\qquad P^{\mathrm{int}}_t=S_t+\varepsilon^{\mathrm{int}},
\qquad 0<|\varepsilon^{(i)}|,|\varepsilon^{\mathrm{int}}|\ll \delta,
\]
where $S$ now denotes the price process well defined under any $\mathbb P^\alpha$,
$\alpha\in\mathcal A$.\\

If there is inventory remaining at $T$, the broker needs to liquidate it immediately. Finally, the cash processes of the investor, the broker and the high frequency trader are given by
\[\text{ (Cash processes) }
\begin{cases}
 \text{(Investor)}& dX_t^{C} =
      (P_t^{\mathrm{int}} u_t - S_t \nu_t) dt
    + \sum_{i=1}^{M} P_t^{(i)}\, dJ_t^{(i),\alpha}, \\[4pt]
 \text{(Broker)}&  dX_t^{B} =
      - (P_t^{\mathrm{lit}} h_t + P_t^{\mathrm{int}} u_t + \lambda \nu_t^2 + \delta|\nu_t|) dt.\\[4pt]
 \text{(HF Traders)}&  dX_t^{H} =
      - P_t^{\mathrm{lit}}\varsigma_t dt.
\end{cases}
\]
From now on, and since we are not considering strategic HF Traders in this section, we enforce the following assumption

\begin{assumption}
    The strategy $\varsigma$ is a feedback control of $Q^H$, that is $\varsigma_t=\varsigma(t,Q_t^H)$. 
\end{assumption}
This assumption is necessary in order to reduce the problem of the broker to an HJB equation and is also justified by the next section and Theorem \ref{thm:mfg_hft_existence} below, by noticing that in a purely competitive setup, with no regulation, the optimal strategy for the HF traders is a function of $t$ and $Q_t^H$ at time $t$.

\subsection{Broker's Optimization Problem}\label{sec:broker}

At terminal time, if $Q^C > 0$, the broker will execute this position immediately with $S_T Q_T^C - \alpha_C (Q_T^C)^2$ and then pay the investor $S_TQ_T^C$, which implies that the penalty for the remaining position is $\alpha_C(Q_T^C)^2$. We moreover fix the optimal compensation $\xi$ as an $\mathcal F_T$-measurable random variable given by the investor to the broker at time $T$. The broker's total terminal payoff is defined by
\[
  \Pi_T^B
  :=
  \xi + X_T^B + Q_T^B\bigl(S_T-\alpha_B Q_T^B\bigr) - \alpha_C (Q_T^C)^2
  - \int_0^T\Bigl(
      \kappa_{\mathrm{tox}} u_t^2
      + \phi (Q_t^B)^2
    \Bigr)dt,
\]
for some constants $\alpha_B,\alpha_C,\kappa_{\mathrm{tox}},\phi\geq 0$.
The broker's value function is then
\[
  V_0^B(\xi)
  :=
  \sup_{(\nu,h,u,\ell)\in\mathcal A}
  \mathbb E^{\mathbb P^\alpha}\!\left[
    \Psi_B\!\left(\Pi_T^B\right)
  \right],
\]
for some non-decreasing and concave utility function $\Psi_B:\mathbb R\longrightarrow\mathbb R$.\\

We denote by $\mathcal C_d$ the set of admissible contracts composed of $\mathcal F_T$-measurable random variables $\xi$ such that
$\sup_{(\nu,h,u,\ell)\in\mathcal A}
\mathbb E^{\mathbb P^\alpha}
\left[
e^{d|\xi|}
\,\middle|\,
\mathcal F_t
\right],$
is a BMO martingale for some $d>0$ large enough specified case by case depending on the choice of $\Psi_B$.
\subsubsection{Linear utility}

We first consider the risk-neutral case
$\Psi_B(x)=x.$
We thus choose $d=1$, so that any contract $\xi\in \mathcal C_1$ is $L^p$-integrable for any $p\geq 1$.
Then the broker's optimization problem becomes
\[
  V_0^{B}(\xi)
  =
  \sup_{(\nu,h,u,\ell)\in\mathcal A}
  \mathbb E^{\mathbb P^\alpha}\!\left[
    \Pi_T^B
  \right],
\]
where
\begin{equation}\label{eq:broker-linear-payoff}
  \Pi_T^B
  :=
  \xi
  + X_T^B
  + Q_T^B\bigl(S_T-\alpha_B Q_T^B\bigr)
  - \alpha_C (Q_T^C)^2
  - \int_0^T\Bigl(
      \kappa_{\mathrm{tox}} u_t^2
      + \phi (Q_t^B)^2
    \Bigr)dt.
\end{equation}

As in our previous linear-utility derivation, we first work with the
terminal random variable itself. Under the admissibility and
integrability assumptions, we directly have the following lemma.
\begin{lemma}\label{lemma:l2}
 The terminal value satisfies,
    $$
    \zeta^{\nu,h,u,\ell} := \xi
  + X_T^B
  + Q_T^B\bigl(S_T-\alpha_B Q_T^B\bigr)
  - \alpha_C (Q_T^C)^2 \in \L^2_T, \quad \forall (\nu,h,u,\ell)\in\mathcal{A}.
    $$   
\end{lemma}
Heuristically, by the martingale representation theorem for the dynamical version of the value function and for every admissible
control $(\nu,h,u,\ell)\in\mathcal A$, we can find a unique triple
\[
  (Y^{\nu,h,u,\ell},Z^{\nu,h,u,\ell},K^{\nu,h,u,\ell})
  \in
  \mathbb S_T^2(\mathbb R)\times
  \mathbb H_T^2(\mathbb R)\times
  \mathbb K_T^2(\mathbb R^M)
\]
such that
\begin{equation}\label{eq:broker-linear-raw-bsde}
  Y_t^{\nu,h,u,\ell}
  =
  \zeta^{\nu,h,u,\ell} - \int_t^T \big(\kappa_{\mathrm{tox}} u_s^2
      + \phi (Q_s^B)^2\big)\, ds
  - \int_t^T Z_s^{\nu,h,u,\ell}\,dW^\alpha_s
  - \sum_{i=1}^M \int_t^T K_s^{i,\nu,h,u,\ell}\,d\widetilde N_s^{(i)}, \quad \P^{\alpha}\text{-a.s.}
\end{equation}
where $\widetilde N^{(i)}$ denotes the compensated Poisson process
associated with dark-pool fills. On the other hand, by applying It\^o's formula, we obtain
\begin{align}
  d\Bigl(X_t^B+Q_t^B(S_t-\alpha_B Q_t^B)\Bigr)
\nonumber  &=
  \Bigl[
      -\varepsilon^{\mathrm{int}}u_t
      -\lambda h_t^2
      -\delta|h_t| -\lambda \nu_t^2
      -\delta|\nu_t|
      -2\alpha_BQ_t^B(u_t+h_t) \\
      &+Q_t^B \Big( (\gamma+\eta)(\nu_t+h_t)+(\epsilon+\eta)\varsigma_t-\kappa I_t \Big)
  \Bigr]dt
  +\sigma Q_t^B\,dW_t.
  \label{eq:broker-b-marked-to-market}
\end{align}

Next we have
\begin{align}
  d\bigl(-\alpha_C(Q_t^C)^2\bigr)
  &=
  -2\alpha_C Q_t^C(\nu_t - u_t)\,dt
  + \sum_{i=1}^M
     \alpha_C \min(\ell_t^i, r^{(i)})\bigl(2Q_{t-}^C-\min(\ell_t^i, r^{(i)})\bigr)\,dN_t^{(i)}.
  \label{eq:qc-penalty-dynamics}
\end{align}

Combining \eqref{eq:broker-b-marked-to-market} and
\eqref{eq:qc-penalty-dynamics}, we obtain
\begin{align}
  d\Bigl(
    X_t^B+Q_t^B(S_t-\alpha_B Q_t^B)-\alpha_C(Q_t^C)^2
  \Bigr)
  \nonumber
  &=
  \Bigl[
      -\varepsilon^{\mathrm{int}}u_t
      -\lambda h_t^2
      -\delta|h_t| -\lambda \nu_t^2
      -\delta|\nu_t|
      -2\alpha_BQ_t^B(u_t+h_t) \\
      &+Q_t^B \Big( (\gamma+\eta)(\nu_t+h_t)+(\epsilon+\eta)\varsigma_t-\kappa I_t \Big)
      -2\alpha_C Q_t^C(\nu_t-u_t)
  \Bigr]dt
  \nonumber\\
  &\quad
  +\sigma Q_t^B\,dW_t
  +\sum_{i=1}^M
    \alpha_C \min(\ell_t^i, r^{(i)})\bigl(2Q_{t-}^C-\min(\ell_t^i, r^{(i)})\bigr)\,dN_t^{(i)}.
  \label{eq:total-marked-to-market}
\end{align}

Combining \eqref{eq:broker-linear-raw-bsde} together with
\eqref{eq:total-marked-to-market} and the definition of $W^\alpha$ we get
\begin{align}
  Y_t^{\nu,h,u,\ell}
  &=
  \xi
  +\int_t^T
     f^B_{\nu,h,u,\ell}(s,Q_s^C,Q_s^B,I_s,Z_s^{\nu,h,u,\ell})\,ds
  \nonumber\\
  &\quad
  -\int_t^T
     \bigl(Z_s^{\nu,h,u,\ell}-\sigma Q_s^B\bigr)\,dW_s
  \nonumber\\
  &\quad
  -\sum_{i=1}^M \int_t^T
     \Bigl(
       K_s^{i,\nu,h,u,\ell}
       -\alpha_C \min(\ell_s^i, r^{(i)})(2Q_{s-}^C-\min(\ell_s^i, r^{(i)}))
     \Bigr)\,d\widetilde N_s^{(i)}, \quad \P^0\text{-a.s.}
  \label{eq:broker-linear-bsde}
\end{align}
where the driver is
\begin{align}
  f^B_{\nu,h,u,\ell}(t,q^C,q^B,\iota,z)
  &=
      -\varepsilon^{\mathrm{int}}u
      -\lambda h^2
      -\delta|h| -\lambda \nu^2
      -\delta|\nu|
      -2\alpha_B q^B(u+h)
        \nonumber\\
  &\quad + \frac{\Bigl((\gamma+\eta)(\nu+h)+(\epsilon+\eta)\varsigma_t-\kappa \iota\Bigr)}{\sigma}z
      - 2\alpha_C q^C(\nu-u)
      -\kappa_{\mathrm{tox}} u^2
      -\phi(q^B)^2
  \nonumber\\
  &\quad
      +\sum_{i=1}^M
        \lambda_0^{(i)}
        \,\alpha_C \min(\ell^i, r^{(i)})\bigl(2q^C-\min(\ell^i, r^{(i)})\bigr).
  \label{eq:broker-linear-driver}
\end{align}

We now formalize it in the following theorem. 

\begin{theorem}
\label{thm:discounted_factor}
For any $(t,q^C,q^B,\iota,z)\in [0,T]\times \mathbb R^2\times \mathbb R^2$ we define  $$
\hat f^B(t,q^C,q^B,\iota,z) = \sup_{\nu,h,u,\ell} f^{B}_{\nu,h,u,\ell}(t,q^C,q^B,\iota,z).
$$ The value function is given by $V_0^B = \hat Y_0$
where $(\hat Y,\hat Z,\hat K)$ is the unique solution to the BSDE \eqref{eq:broker-linear-bsde} with driver $\hat f^B$ and terminal condition $\xi$,
and the corresponding optimal trading strategies $(\hat\nu,\hat h,\hat u,\hat \ell)$ satisfy
\begin{enumerate}
    \item[(a)]
    \begin{enumerate}
        \item[1.] If $M=1$, $\hat \ell_t = Q^C_t$.
        \item[2.] If $M>1$, $\lambda_0^{(i)}(Q_t^C - \hat\ell^{i}_t)(1 - F_i(\hat\ell^{i}_t)) = \lambda_0^{(j)}(Q_t^C - \hat\ell^{j}_t)(1 - F_j(\hat\ell^{j}_t)), \, \forall i,j$.
    \end{enumerate}
    \item[(b)]
$\hat \nu_t=\nu^*(Q_t^C,\hat Z_t)$ where
\[ \nu^*(q^C,z)=-\frac{1}{2\lambda}
\Biggl[
\left(2\alpha_C q^C-\frac{\gamma+\eta}{\sigma}z-\delta\right)
\mathbf{1}_{\left\{\,2\alpha_C q^C-\frac{\gamma+\eta}{\sigma}z>\delta\,\right\}}
+
\left(2\alpha_C q^C-\frac{\gamma+\eta}{\sigma}z+\delta\right)
\mathbf{1}_{\left\{\,2\alpha_C q^C-\frac{\gamma+\eta}{\sigma}z<-\delta\,\right\}}
\Biggr].
\]
    \item[(c)]
$\hat h_t=h^*(Q_t^B,\hat Z_t) $ where 

\[h^*(q^B,z)
=
-\frac{1}{2\lambda}
\Biggl[
\left(2\alpha_B q^B-\frac{\gamma+\eta}{\sigma}z-\delta\right)
\mathbf{1}_{\left\{\,2\alpha_B q^B-\frac{\gamma+\eta}{\sigma}z>\delta\,\right\}}
+
\left(2\alpha_B q^B-\frac{\gamma+\eta}{\sigma}z+\delta\right)
\mathbf{1}_{\left\{\,2\alpha_B q^B-\frac{\gamma+\eta}{\sigma}z<-\delta\,\right\}}
\Biggr].\]
    \item[(d)]
    $
\hat u_t=u^*(Q_t^B,Q_t^C)$ where
\[u^*(q^B,q^C)=
\Big(\frac{-\varepsilon^{\mathrm{int}}-2\alpha_B q^B+2\alpha_C q^C}{2\kappa_{\mathrm{tox}}}\Big)^+.
\]
    
\end{enumerate}

\end{theorem}

\begin{proof}
    For any $(\nu, h, u, \ell) \in \mathcal{A}$, we recall from Lemma \ref{lemma:l2} that the terminal value $\zeta^{\nu,h,u,\ell} \in \L^2_T$. Hence,
     $\E[\int_0^T f^B_{\nu,h,u,\ell}(t,Q_t^C,Q_t^B,I_t,0)\, dt] < \infty$.
    Furthermore, from the construction of BSDE \eqref{eq:broker-linear-bsde}, we have $(Z^{\nu,h,u,\ell},K^{\nu,h,u,\ell})\in \H^2_T(\R)\times \K^2_T(\R^M)$. Therefore we deduce from Theorem 4.5 in \cite{oksendal2007applied}, that there exists a unique $(Y^{\nu,h,u,\ell},Z^{\nu,h,u,\ell},K^{\nu,h,u,\ell})\in \S^2_T(\R) \times \H^2_T(\R)\times \K^2_T(\R^M)$ solution to BSDE \eqref{eq:broker-linear-bsde}. From the comparison theorem for BSDE \eqref{eq:broker-linear-bsde}, see for example \cite{antonelli2016solutions}, we deduce that $Y_0^{\nu,h,u,\ell} \leq Y_0^{\hat \nu,\hat h,\hat u,\hat \ell}$ for all $(\nu,h,u,\ell) \in \mathcal{A}$. Hence, we deduce that $V_0^B = \sup_{\nu,h,u,\ell} Y_0^{\nu,h,u,\ell}=Y_0^{\hat \nu,\hat h,\hat u,\hat \ell}$. 
    To solve the optimal strategy, we note that the optimization can be decomposed into four parts,
    \begin{align*}
        \sup_{\nu,h,u,\ell} f^B_{\nu,h,u,\ell}(t,q^C,q^B,\iota,z)
        &=
        -\phi (q^B)^2
        + \frac{(\epsilon+\eta)\varsigma_t-\kappa \iota}{\sigma}\,z \\
        &\quad
        + \sup_{\nu}\left\{
        -\lambda \nu^2 - \delta|\nu|
        -2\alpha_C q^C \nu
        + \frac{\gamma+\eta}{\sigma} z\,\nu
        \right\} \\
        &\quad
        + \sup_{h}\left\{
        -\lambda h^2 - \delta|h|
        -2\alpha_B q^B h
        + \frac{\gamma+\eta}{\sigma} z\,h
        \right\} \\
        &\quad
        + \sup_{u}\left\{
        -\varepsilon^{\mathrm{int}}u
        -2\alpha_B q^B u
        +2\alpha_C q^C u
        -\kappa_{\mathrm{tox}}u^2
        \right\} \\
        &\quad
        + \sup_{\ell}
        \sum_{i=1}^M \lambda_0^{(i)}\alpha_C
        \mathbb E\!\left[
        \min(\ell^i,r^{(i)})\bigl(2q^C-\min(\ell^i,r^{(i)})\bigr)
        \right] \\
        &\qquad \text{s.t. } \sum_{i=1}^M \ell^i \le q^C .
        \end{align*}
        For the dark pool trading strategy, the corresponding Lagrangian is given by \begin{align*}
        L(\ell, \vartheta) &= \sum_{i=1}^M \lambda_0^{(i)} \E \left[ \min(\ell^i, r^{(i)}) \left( 2q^C - \min(\ell^i, r^{(i)}) \right) \right] - \vartheta \left(\sum_{i=1}^M \ell^i - q^C\right) \\
    &= \sum_{i=1}^M\lambda_0^{(i)} \ell^i(2q^C - \ell^i)(1-F_i(\ell^i)) - \vartheta \left(\sum_{i=1}^M \ell^i - q^C\right).
    \end{align*}
    By optimality conditions, we have
    \begin{align*}
    \nabla_\ell L(\ell, \vartheta) &= \begin{pmatrix}
       2\lambda_0^{(i)}(q^C-\ell^{i})(1-F_i(\ell^i)) - \vartheta \\
        \vdots \\
        2\lambda_0^{(j)}(q^C-\ell^{j})(1-F_j(\ell^j)) - \vartheta
        \end{pmatrix} \\
        &= 0.
    \end{align*}
The equalization condition of (a) then follows directly. The optimal lit-market trading and internalization strategies follow from the first-order conditions of the corresponding sub-problems.
\end{proof}

\subsubsection{Exponential utility}
\label{sec:expo}
We now consider a risk-averse broker with exponential utility
\[
\Psi_B(x) = -e^{-\rho x}, \qquad \rho > 0.
\]

The broker's optimization problem becomes
\[
V_0^B = \sup_{(\nu,h,u,\ell)\in\mathcal A}
\mathbb E^{\mathbb P^\alpha}\Big[
- e^{-\rho \Pi_T^B}
\Big],
\]
where
\[
\Pi_T^B
=
\xi
+ X_T^B
+ Q_T^B(S_T - \alpha_B Q_T^B)
- \alpha_C (Q_T^C)^2
- \int_0^T \big(
\kappa_{\mathrm{tox}} u_t^2 + \phi (Q_t^B)^2
\big) dt.
\]

Using the same heuristic and computations as in the linear utility case, we rewrite
\begin{align}
\Pi_T^B
&=
\xi + X_0^B - \alpha_C (Q_0^C)^2
+ \int_0^T \Big(
-\lambda \nu_t^2 - \delta|\nu_t|
-\lambda h_t^2 - \delta|h_t|
- \varepsilon^{\mathrm{int}} u_t
\nonumber\\
&\quad
- 2\alpha_C Q_t^C \nu_t
- 2\alpha_B Q_t^B h_t
- 2\alpha_B Q_t^B u_t
+ 2\alpha_C Q_t^C u_t
- \kappa_{\mathrm{tox}} u_t^2
- \phi (Q_t^B)^2
\Big) dt
\nonumber\\
&\quad
+ \int_0^T \sigma Q_t^B dW_t
+ \sum_{i=1}^M \int_0^T
\alpha_C \min(\ell_t^{(i)}, r^{(i)})
\big(2Q_{t-}^C - \min(\ell_t^{(i)}, r^{(i)})\big)
dN_t^{(i)}.
\label{eq:exp-pnl}
\end{align}

We consider the BSDE
\begin{equation}
dY_t
=
Z_t dW_t
+ \sum_{i=1}^M K_t^{i} d\widetilde N_t^{(i)}
- H^B(t,Z_t,K_t) dt,
\quad Y_T = \xi,
\label{eq:exp-bsde}
\end{equation}
where \(\widetilde N^{(i)}\) are compensated Poisson processes. The Hamiltonian is defined by
\[
H^B(t,z,k)
=
\sup_{(\nu,h,u,\ell)\in\mathcal A}
h^B_{\nu,h,u,\ell}(t,z,k),
\]
with
\begin{align}
h_{\nu,h,u,\ell}^B
&=
-\varepsilon^{\mathrm{int}}u
-\lambda h^2-\delta|h|
-\lambda \nu^2-\delta|\nu|
-2\alpha_B q^B(u+h)
-2\alpha_C q^C(\nu-u)
\nonumber\\
&\quad
+\frac{(\gamma+\eta)(\nu+h)+(\epsilon+\eta)\varsigma_t-\kappa I_t}{\sigma}\,z
\nonumber\\
&\quad
-\kappa_{\mathrm{tox}}u^2
-\phi(q^B)^2
-\frac{\rho}{2}(z+\sigma q^B)^2
-\sum_{i=1}^M
\frac{\lambda_0^{(i)}}{\rho}
\left(
e^{-\rho k_i}
\mathbb E\!\left[
e^{-\rho \alpha_C \min(\ell^i,r^{(i)})\bigl(2q^C-\min(\ell^i,r^{(i)})\bigr)}
\right]
+\rho k_i-1
\right)
\label{eq:exp-driver}
\end{align}
The second line transfers the drift of $S$ (and hence of $Q^B S$) into the BSDE through Girsanov, paralleling the same term in the linear-utility driver \eqref{eq:broker-linear-driver}. Without it, the contributions of the broker's own footprint $(\gamma+\eta)(\nu+h)$, the HFT flow $(\epsilon+\eta)\varsigma_t$, and the transient-impact decay $-\kappa I_t$ to the mark-to-market $\mathrm{d}(Q^B_t S_t)$ would not appear in the Hamiltonian.

\begin{assumption}\label{assump:bmo_martingale}
We consider an additional integrability assumption on the compensation
\(\xi\) and assume furthermore that
\[
\sup_{(\nu,h,u,\ell)\in\mathcal A}
\mathbb E^{\mathbb P^{\alpha}}
\left[
e^{d|\xi|}
\,\middle|\,
\mathcal F_t
\right]
\]
is a BMO martingale for some \(d>2\rho\). In the sequel, we simply write $\mathcal C$ for the set of contracts $\xi\in\mathcal C_d$ satisfying this assumption.
\end{assumption}

\begin{lemma}\label{lemma:bsde_y_existence}
Under Assumption \ref{assump:bmo_martingale}, there exists a triple
\[
(Y,Z,K)\in \mathbb{S}^2_T(\mathbb{R}) \times \mathbb{H}^2_T(\mathbb{R}) \times \mathbb{K}^2_T(\mathbb{R}^M)
\]
solution to the BSDE \eqref{eq:exp-bsde}.
\end{lemma}

\begin{proof}
We analyze the driver of BSDE \eqref{eq:exp-bsde}. Define
\[
[k]_{-\rho}
:=
\sum_{i=1}^M
\frac{e^{-\rho k_i}-(-\rho k_i)-1}{-\rho}\lambda_0^{(i)}.
\]

Since
\[
\alpha_C \min(\ell^i, r^{(i)})\bigl(2q^C - \min(\ell^i, r^{(i)})\bigr)\ge 0
\quad \text{a.s.},
\]
we have
\[
-\sum_{i=1}^M
\frac{\lambda_0^{(i)}}{\rho}
\left(
e^{-\rho k_i}
\mathbb{E}\!\left[
e^{-\rho \alpha_C \min(\ell^i, r^{(i)})\bigl(2q^C-\min(\ell^i, r^{(i)})\bigr)}
\right]
+\rho k_i-1
\right)
\le [k]_{-\rho}.
\]

Next, consider the quadratic term
\[
-\frac{\rho}{2}(z+\sigma q^B)^2
=
-\frac{\rho}{2}z^2 - \rho \sigma q^B z - \frac{\rho}{2}\sigma^2 (q^B)^2.
\]

Applying Young's inequality,
\[
|\rho \sigma q^B z|
\le
\frac{\rho}{2}z^2 + \frac{1}{2\rho}(\rho \sigma q^B)^2.
\]

Therefore, the driver satisfies
\begin{align*}
|h^B(t,z,k)|
&\le
\rho z^2 + [k]_{-\rho} + \varpi,
\end{align*}
where
\[
\varpi
=
\frac{1}{2\rho}(\rho \sigma q^B)^2
+
\bigg|
-\varepsilon^{\mathrm{int}}u
-\lambda h^2 - \delta|h|
-\lambda \nu^2 - \delta|\nu|
-2\alpha_B q^B(u+h)
-2\alpha_C q^C(\nu-u)
-\kappa_{\mathrm{tox}}u^2
-\phi (q^B)^2
\bigg|.
\]

Hence,
\[
-\varpi - \rho z^2 - [-k]_{-\rho}
\le
h^B(t,z,k)
\le
\varpi + \rho z^2 + [k]_{-\rho}.
\]

Next, for the jump component, by the mean value theorem, there exists
\(\beta_t^i > 0\) such that
\begin{align*}
h^B(t,z,k') - h^B(t,z,k)
&=
-\sum_{i=1}^M \lambda_0^{(i)}\big((k')^i-k^i\big)
-\sum_{i=1}^M \frac{\lambda_0^{(i)}}{\rho}
\big(e^{-\rho (k')^i} - e^{-\rho k^i}\big)
\mathbb{E}\!\left[e^{-\rho \alpha_C \min(\ell^i,r^{(i)})\bigl(2q^C-\min(\ell^i,r^{(i)})\bigr)}\right] \\
&\leq
\sum_{i=1}^M
\lambda_0^{(i)}(\beta_t^i-1)\big((k')^i-k^i\big),
\end{align*}

Since \(\beta_t^i>0\), there exist constants \(D_1\in(-1,0]\) and \(D_2>0\)
such that
\[
\beta_t^i-1 \in [D_1,D_2], \quad \forall i.
\]

Finally, for the \(z\)-dependence,
\begin{align*}
|h^B(t, z,k) - h^B(t, z',k)| &= |- \frac{1}{2}\rho (z+q^B\sigma)^2 + \frac{1}{2}\rho (z' + q^B\sigma)^2| \\
        &\leq \frac{\rho}{2}|z^2 - (z')^2| + |\rho q^B\sigma | |z - z'| \\
        &= \frac{\rho}{2}|z - z'||z + z'| + |\rho q^B\sigma | |z - z'| \\
        &= \Big( \frac{\rho}{2}\big(|z| + |z'|\big) + |\rho q^B\sigma | \Big)|z - z'|.
\end{align*}

Therefore, the driver has quadratic growth in \(z\) and exponential structure
in the jump component. By \cite[Theorem 3]{antonelli2016solutions},
there exists a solution
\[
(Y,Z,K)\in \mathbb{S}^2_T(\mathbb R) \times \mathbb{H}^2_T(\mathbb R) \times \mathbb{K}^2_T(\mathbb R^M)
\]
to BSDE \eqref{eq:exp-bsde}.
\end{proof}

While the BSDE involved in the problem is quadratic with jumps, the article \cite{antonelli2016solutions} does not provide uniqueness and comparison results. For the sake of consistency, we recall and prove a comparison result for our quadratic BSDE with jumps in our context.

\begin{lemma}[Comparison theorem]\label{lemma:comparison}
Suppose Assumption \ref{assump:bmo_martingale} holds. Let
\[
(Y^1,Z^1,K^1), \qquad (Y^2,Z^2,K^2)
\]
be two solutions in
\[
\mathbb S_T^2(\mathbb R)\times \mathbb H_T^2(\mathbb R)\times \mathbb K_T^2(\mathbb R^M)
\]
to BSDE \eqref{eq:exp-bsde}, with terminal conditions \(\xi^1,\xi^2\in\mathcal C\) and respective generators
\[
h^{B,1}(t,z,k), \qquad h^{B,2}(t,z,k).
\]
Assume moreover that
\[
h^{B,1}(t,Z_t^1,K_t^1)\le h^{B,2}(t,Z_t^1,K_t^1),
\qquad \forall t\in[0,T], \quad \mathbb P^0\text{-a.s.},
\]
and
\[
\xi^1\le \xi^2,\qquad \mathbb P^0\text{-a.s.}
\]
Then
\[
Y_t^1\le Y_t^2,\qquad \forall t\in[0,T], \quad \mathbb P^0\text{-a.s.}
\]
\end{lemma}

\begin{proof}
First, we introduce an auxiliary process
\[
\widetilde Y_t=-e^{-\rho Y_t}.
\]
Then It\^o's formula with jumps gives
\begin{align}\label{eq:tilde_bsde}
\widetilde Y_t
=
\widetilde \xi
+
\int_t^T
\widetilde h^B(s,\widetilde Y_s,\widetilde Z_s,\widetilde K_s)\,ds
-
\int_t^T \widetilde Z_s\,dW_s
-
\sum_{i=1}^M \int_t^T \widetilde K_s^i\,d\widetilde N_s^{(i)},
\quad \mathbb P^0\text{-a.s.},
\end{align}
where
\[
\widetilde \xi=-e^{-\rho \xi},
\qquad
\widetilde Z_t=-\rho \widetilde Y_{t^-}Z_t,
\qquad
\widetilde K_t^i=\widetilde Y_{t^-}\bigl(e^{-\rho K_t^i}-1\bigr).
\]

Recall that the driver of BSDE \eqref{eq:exp-bsde} is
\begin{align*}
h^B(t,z,k)
&=
-\varepsilon^{\mathrm{int}}u
-\lambda h^2
-\delta|h|
-\lambda \nu^2
-\delta|\nu|
-2\alpha_B q^B(u+h)
-2\alpha_C q^C(\nu-u) \\
&\quad
-\kappa_{\mathrm{tox}}u^2
-\phi(q^B)^2
-\frac{\rho}{2}(z+\sigma q^B)^2 \\
&\quad
-\sum_{i=1}^M \frac{\lambda_0^{(i)}}{\rho}
\left(
e^{-\rho k_i}
\mathbb E\!\left[
e^{-\rho \alpha_C \min(\ell^i,r^{(i)})\bigl(2q^C-\min(\ell^i,r^{(i)})\bigr)}
\right]
+\rho k_i-1
\right).
\end{align*}

Applying It\^o's formula to \(\widetilde Y_t=-e^{-\rho Y_t}\), we obtain the driver function
\begin{align*}
\widetilde h^B(t,\widetilde y,\widetilde z,\widetilde k)
&=
\widetilde y\Big(
\rho \varepsilon^{\mathrm{int}}u
+\rho \lambda h^2
+\rho \delta|h|
+\rho \lambda \nu^2
+\rho \delta|\nu|
+2\rho \alpha_B q^B(u+h)
+2\rho \alpha_C q^C(\nu-u) \\
&\qquad\qquad
+\rho \kappa_{\mathrm{tox}}u^2
+\rho \phi(q^B)^2
+\frac{1}{2}\rho^2 \sigma^2 (q^B)^2
\Big)
+\widetilde z(-\rho \sigma q^B) \\
&\quad
+\sum_{i=1}^M \lambda_0^{(i)}(\widetilde k^i+\widetilde y)
\left(
\mathbb E\!\left[
e^{-\rho \alpha_C \min(\ell^i,r^{(i)})\bigl(2q^C-\min(\ell^i,r^{(i)})\bigr)}
\right]-1
\right).
\end{align*}
We note that $\widetilde{Y}$ solves a linear BSDE with terminal condition $\widetilde{\xi}\in \L_T^2$. Then by \cite[Theorem 4.3.1]{zhang2017backward}, there exists a unique solution to BSDE \eqref{eq:tilde_bsde}. By definition, we have
\[
\widetilde Y_t=-e^{-\rho Y_t},
\qquad
\widetilde \xi=-e^{-\rho \xi},
\]
which implies
$
\widetilde \xi^1\le \widetilde \xi^2.
$ Furthermore,
\begin{align*}
&\widetilde h^{B,1}(t,\widetilde Y_t^1,\widetilde Z_t^1,\widetilde K_t^1)
-\widetilde h^{B,2}(t,\widetilde Y_t^1,\widetilde Z_t^1,\widetilde K_t^1)\\
&=
\Bigg(
-\rho \widetilde Y_t^1 h^{B,1}(t,Z_t^1,K_t^1)
+\widetilde Y_t^1\sum_{i=1}^M \lambda_0^{(i)}
\bigl(
-\rho K_t^{1,i}+1-e^{-\rho K_t^{1,i}}
\bigr)
-\frac{1}{2}\rho^2 (Z_t^1)^2 \widetilde Y_t^1
\Bigg)\\
&\quad
-\Bigg(
-\rho \widetilde Y_t^1 h^{B,2}(t,Z_t^1,K_t^1)
+\widetilde Y_t^1\sum_{i=1}^M \lambda_0^{(i)}
\bigl(
-\rho K_t^{1,i}+1-e^{-\rho K_t^{1,i}}
\bigr)
-\frac{1}{2}\rho^2 (Z_t^1)^2 \widetilde Y_t^1
\Bigg)\\
&=
-\rho \widetilde Y_t^1
\Big(
h^{B,1}(t,Z_t^1,K_t^1)-h^{B,2}(t,Z_t^1,K_t^1)
\Big)\\
&\le 0,
\end{align*}
recalling that \(\widetilde Y_t^1<0\) and
\(h^{B,1}(t,Z_t^1,K_t^1)\le h^{B,2}(t,Z_t^1,K_t^1)\). Let
\[
\widehat Y_t=\widetilde Y_t^2-\widetilde Y_t^1,\qquad
\widehat Z_t=\widetilde Z_t^2-\widetilde Z_t^1,\qquad
\widehat K_t=\widetilde K_t^2-\widetilde K_t^1.
\]
Then we get
\begin{align*}
d\widehat Y_t
=
-\Big(
\widetilde h^{B,2}(t,\widetilde Y_t^2,\widetilde Z_t^2,\widetilde K_t^2)
-
\widetilde h^{B,1}(t,\widetilde Y_t^1,\widetilde Z_t^1,\widetilde K_t^1)
\Big)\,dt
+\widehat Z_t\,dW_t
+\sum_{i=1}^M \widehat K_t^i\,d\widetilde N_t^{(i)}.
\end{align*}
Therefore,
\begin{align*}
\widetilde h^{B,2}(t,\widetilde Y_t^2,\widetilde Z_t^2,\widetilde K_t^2)
-
\widetilde h^{B,1}(t,\widetilde Y_t^1,\widetilde Z_t^1,\widetilde K_t^1)
&=
\widetilde h^{B,2}(t,\widetilde Y_t^2,\widetilde Z_t^2,\widetilde K_t^2)
-
\widetilde h^{B,2}(t,\widetilde Y_t^1,\widetilde Z_t^2,\widetilde K_t^2) \\
&\quad
+
\widetilde h^{B,2}(t,\widetilde Y_t^1,\widetilde Z_t^2,\widetilde K_t^2)
-
\widetilde h^{B,2}(t,\widetilde Y_t^1,\widetilde Z_t^1,\widetilde K_t^2) \\
&\quad
+
\widetilde h^{B,2}(t,\widetilde Y_t^1,\widetilde Z_t^1,\widetilde K_t^2)
-
\widetilde h^{B,2}(t,\widetilde Y_t^1,\widetilde Z_t^1,\widetilde K_t^1) \\
&\quad
+
\widetilde h^{B,2}(t,\widetilde Y_t^1,\widetilde Z_t^1,\widetilde K_t^1)
-
\widetilde h^{B,1}(t,\widetilde Y_t^1,\widetilde Z_t^1,\widetilde K_t^1).
\end{align*}

Since \(\widetilde h^{B,2}\) is linear in \((\widetilde y,\widetilde z,\widetilde k)\), we have
\[
\widetilde h^{B,2}(t,\widetilde Y_t^2,\widetilde Z_t^2,\widetilde K_t^2)
-
\widetilde h^{B,1}(t,\widetilde Y_t^1,\widetilde Z_t^1,\widetilde K_t^1)
=
\chi(t)+\phi^y(t) \widehat Y_t+\phi^z(t) \widehat Z_t+\sum_{i=1}^M \phi^{k,i}(t)\widehat K_t^i,
\]
where
\[
\chi(t)
:=
\widetilde h^{B,2}(t,\widetilde Y_t^1,\widetilde Z_t^1,\widetilde K_t^1)
-
\widetilde h^{B,1}(t,\widetilde Y_t^1,\widetilde Z_t^1,\widetilde K_t^1)
\ge 0,
\]
and
\begin{align*}
\phi^y(t)
&=
\rho \varepsilon^{\mathrm{int}}u^2_t
+\rho \lambda (h_t^2)^2
+\rho \delta|h_t^2|
+\rho \lambda (\nu_t^2)^2
+\rho \delta|\nu_t^2|
+2\rho \alpha_B q_t^B(u_t^2+h_t^2) \\
&\quad
+2\rho \alpha_C q_t^C(\nu_t^2-u_t^2)
+\rho \kappa_{\mathrm{tox}}(u_t^2)^2
+\rho \phi(q_t^B)^2
+\frac{1}{2}\rho^2\sigma^2(q_t^B)^2 \\
&\quad
+\sum_{i=1}^M \lambda_0^{(i)}
\left(
\mathbb E\!\left[
e^{-\rho \alpha_C \min(\ell_t^{2,i},r^{(i)})\bigl(2q_t^C-\min(\ell_t^{2,i},r^{(i)})\bigr)}
\right]-1
\right),\\
\phi^z(t)
&=
-\rho \sigma q_t^B,\\
\phi^{k,i}(t)
&=
\lambda_0^{(i)}
\left(
\mathbb E\!\left[
e^{-\rho \alpha_C \min(\ell_t^{2,i},r^{(i)})\bigl(2q_t^C-\min(\ell_t^{2,i},r^{(i)})\bigr)}
\right]-1
\right).
\end{align*}

Combining the equations above, we obtain
\begin{align*}
d\widehat Y_t
&=
-\widehat h(t) dt
+\widehat Z_t\,dW_t
+\sum_{i=1}^M \widehat K_t^i\,d\widetilde N_t^{(i)},
\end{align*}
with terminal condition
\[
\widehat Y_T=\widetilde \xi^2-\widetilde \xi^1\ge 0.
\]

Note that
\(
\widehat h(t)
\ge
\phi^y(t) \widehat Y_t+\phi^z(t) \widehat Z_t+\sum_{i=1}^M \phi^{k,i}(t)\widehat K_t^i.
\)
Then by Corollary 4.11 in \cite{oksendal2007applied}, we deduce that
\[
\widetilde Y_t^1\le \widetilde Y_t^2.
\]

Finally, since the function
\[
x\mapsto -\frac{\log(-x)}{\rho}
\]
is increasing on \((-\infty,0)\), we conclude that
$
Y_t^1\le Y_t^2.$
\end{proof}

As a direct consequence of this comparison theorem, we obtain uniqueness.

\begin{corollary}\label{corol:unique_y}
    Under Assumption \ref{assump:bmo_martingale} with $\xi\in \mathcal C$, the solution to BSDE \eqref{eq:exp-bsde} is unique.
\end{corollary}
We now introduce an alternative set of compensation schemes, defined through the BSDE \eqref{eq:exp-bsde} above:
\begin{align*}
\Xi &:= \Big\{ \xi^{Y_0, Z, K} = Y_0 - \int_0^T H^B(s, Z_s, K_s) ds + \int_0^T Z_s dW_s + \sum_{i=1}^M \int_0^T K^{i}_s d\widetilde{N}^{(i)}_s,\\
&\hspace{1.5em} \text{with } Z\in \H^2_T(\R),\; K \in \K_T^2(\R^M),\; -e^{-\rho Y_0}\geq R_0 \Big\}.
\end{align*}
Note that, unlike the make-take fee formulation of \cite{euch2021optimal,baldacci2023market}, the broker's Hamiltonian $h^B_{\nu,h,u,\ell}$ in \eqref{eq:exp-driver} is unconditionally strictly concave in $\nu$ (with $\partial_{\nu\nu}h^B=-2\lambda<0$), and likewise in $h$ and $u$, so the broker's maximizer $(\hat\nu,\hat h,\hat u,\hat\ell)$ exists and is unique without any further restriction on the jump-compensation process $K$.

\begin{theorem}\label{thm:bsde_value}
\begin{enumerate}
    \item Any contract $\xi \in \mathcal{C}$ has a unique representation as 
    $\xi = \xi^{Y_0,Z,K}$ for some $(Y,Z,K)\in \mathbb{S}^2_T(\mathbb{R})
    \times \mathbb{H}^2_T(\mathbb{R})\times \mathbb{K}^2_T(\mathbb{R}^M)$.
    In particular, $\mathcal{C} = \Xi$.

    \item For any $\xi = \xi^{Y_0,Z,K}\in \Xi$, the broker's value function is
    $V^B_0 = -e^{-\rho\bar{Y}_0}$ and the 
    optimal controls $(\hat\nu,\hat h,\hat u,\hat\ell)\in\mathcal{A}$ are:
    \begin{enumerate}
        \item[(a)] \begin{enumerate}
            \item[1.] If $M=1$, $\hat\ell_t = Q^C_t$.
            \item[2.] If $M>1$,
            \begin{align*}
                &\lambda_0^{(i)}\bigl(Q^C_t-\hat\ell^{\,i}_t\bigr)
                  \exp\!\Bigl(-\rho\alpha_C\hat\ell^{\,i}_t
                  \bigl(2Q^C_t-\hat\ell^{\,i}_t\bigr)\Bigr)
                  \bigl(1-F_i(\hat\ell^{\,i}_t)\bigr)\\
                &=\lambda_0^{(j)}\bigl(Q^C_t-\hat\ell^{\,j}_t\bigr)
                  \exp\!\Bigl(-\rho\alpha_C\hat\ell^{\,j}_t
                  \bigl(2Q^C_t-\hat\ell^{\,j}_t\bigr)\Bigr)
                  \bigl(1-F_j(\hat\ell^{\,j}_t)\bigr),
                \quad\forall\,i,j\le M,\;\forall\,t\in[0,T].
            \end{align*}
        \end{enumerate}

        \item[(b)] \(
          \hat\nu_t = -\frac{1}{2\lambda}
          \Biggl[
            \Bigl(2\alpha_C Q^C_t - \tfrac{\gamma+\eta}{\sigma}Z_t - \delta\Bigr)
            \mathbf{1}_{\bigl\{2\alpha_C Q^C_t-\frac{\gamma+\eta}{\sigma}Z_t > \delta\bigr\}}
          + \Bigl(2\alpha_C Q^C_t - \tfrac{\gamma+\eta}{\sigma}Z_t + \delta\Bigr)
            \mathbf{1}_{\bigl\{2\alpha_C Q^C_t-\frac{\gamma+\eta}{\sigma}Z_t < -\delta\bigr\}}
          \Biggr].
        \)

        \item[(c)] \(
          \hat h_t = -\frac{1}{2\lambda}
          \Biggl[
            \Bigl(2\alpha_B Q^B_t - \tfrac{\gamma+\eta}{\sigma}Z_t - \delta\Bigr)
            \mathbf{1}_{\bigl\{2\alpha_B Q^B_t-\frac{\gamma+\eta}{\sigma}Z_t > \delta\bigr\}}
          + \Bigl(2\alpha_B Q^B_t - \tfrac{\gamma+\eta}{\sigma}Z_t + \delta\Bigr)
            \mathbf{1}_{\bigl\{2\alpha_B Q^B_t-\frac{\gamma+\eta}{\sigma}Z_t < -\delta\bigr\}}
          \Biggr].
        \)

        \item[(d)]
        \(
          \hat u_t = \left(
            \frac{-\varepsilon^{\mathrm{int}} - 2\alpha_B Q^B_t + 2\alpha_C Q^C_t}
                 {2\kappa_{\mathrm{tox}}}
          \right)^{\!\!+}.
        \)
    \end{enumerate}
\end{enumerate}
\end{theorem}

\begin{proof}
We start by defining
\begin{equation}\label{eq:Ybar_def}
    \bar{Y}_t := Y_t + X^B_t + Q^B_t(S_t - \alpha_B Q^B_t) 
    - \alpha_C(Q^C_t)^2 - \int_0^t \bigl(\kappa_{\mathrm{tox}}u_s^2 
    + \phi(Q^B_s)^2\bigr)\,ds.
\end{equation}
Applying It\^o's formula to $e^{-\rho\bar{Y}_t}$, we have
\begin{align*}
    \frac{de^{-\rho\bar{Y}_t}}{\rho e^{-\rho\bar{Y}_{t-}}}
    &= \Bigl[H^B(t,Z_t,K_t) - \Bigl(-\varepsilon^{\mathrm{int}}u
    - \lambda h^2 - \delta|h| - \lambda\nu^2 - \delta|\nu|
    - 2\alpha_B Q^B_t(u+h) - 2\alpha_C Q^C_t(\nu - u)\\
    &\quad + \tfrac{(\gamma+\eta)(\nu+h) + (\epsilon+\eta)\varsigma_t - \kappa I_t}{\sigma}\,Z_t \\
    &\quad - \kappa_{\mathrm{tox}}u^2 - \phi(Q^B_t)^2
    - \tfrac{\rho}{2}(Z_t + \sigma Q^B_t)^2
    - \sum_{i=1}^M \frac{\lambda_0^{(i)}}{\rho}
    \Bigl(e^{-\rho K^i_t}\,\mathbb{E}\!\left[e^{-\rho\alpha_C
    \min(\ell^i_t,r^{(i)})(2Q^C_{t-}-\min(\ell^i_t,r^{(i)}))}\right]
    + \rho K^i_t - 1\Bigr)\Bigr)\Bigr]dt \\
    &\quad - (Z_t + \sigma Q^B_t)\,dW_t 
    + \sum_{i=1}^M \frac{1}{\rho}
    \Bigl(e^{-\rho\left(K^i_t + \alpha_C\min(\ell^i_t,r^{(i)})
    (2Q^C_{t-}-\min(\ell^i_t,r^{(i)}))\right)} - 1\Bigr)
    d\widetilde{N}^{(i)}_t.
\end{align*}
Since $H^B(t,z,k) = \sup_{\nu,h,u,\ell} h^B_{\nu,h,u,\ell}(t,z,k)$, 
we deduce that
\begin{align*}
    V^B_0 &= \sup_{(\nu,h,u,\ell)\in\mathcal{A}} 
    \mathbb{E}^{\mathbb P^{\alpha}}\!\left[-e^{-\rho\Pi^B_T}\right] \\
    &= \sup_{(\nu,h,u,\ell)\in\mathcal{A}} 
    \mathbb{E}^{\mathbb P^0}\!\left[-e^{-\rho\bar{Y}_T}\mathcal{E}^{\nu}_T D_T\right] \\
    &= -e^{-\rho\bar{Y}_0} - \mathbb{E}^{\mathbb P^0}\!\left[
    \mathcal{E}^{\nu}_T D_T \int_0^T \rho e^{-\rho\bar{Y}_t}
    \bigl(H^B(t,Z_t,K_t) - h^B_{\nu,h,u,\ell}(t,Z_t,K_t)\bigr)dt
    \right] \\
    &\leq -e^{-\rho\bar{Y}_0}, \quad \forall\,(\nu,h,u,\ell)\in\mathcal{A},
\end{align*}
and equality holds if and only if $(\nu_t,h_t,u_t,\ell_t)$ achieves the 
supremum of $H^B(t,Z_t,K_t)$ for a.e.\ $t$, $\mathbb{P}^0$-a.s. 
Hence $V^B_0 = -e^{-\rho\bar{Y}_0}$. To solve for the optimal controls, we note that the optimization in 
$\sup_{\nu,h,u,\ell} h^B_{\nu,h,u,\ell}$ decomposes into four separable 
sub-problems:
\begin{align*}
    \sup_{\nu,h,u,\ell} h^B_{\nu,h,u,\ell}(t,z,k) 
    &= -\phi(q^B)^2 + \frac{(\epsilon+\eta)\varsigma_t - \kappa I_t}{\sigma}z \\
    &\quad + \sup_{\nu}\Bigl\{-\lambda\nu^2 - \delta|\nu| 
    - 2\alpha_C q^C \nu + \tfrac{\gamma+\eta}{\sigma}z\,\nu\Bigr\} \\
    &\quad + \sup_{h}\Bigl\{-\lambda h^2 - \delta|h| 
    - 2\alpha_B q^B h + \tfrac{\gamma+\eta}{\sigma}z\,h\Bigr\} \\
    &\quad + \sup_{u \geq 0}\Bigl\{-\varepsilon^{\mathrm{int}}u 
    - 2\alpha_B q^B u + 2\alpha_C q^C u - \kappa_{\mathrm{tox}}u^2\Bigr\} \\
    &\quad + \sup_{\ell:\,\sum_i \ell^i \leq q^C}
    \sum_{i=1}^M \frac{\lambda_0^{(i)}}{\rho}
    \Bigl(-e^{-\rho k^i}\,\mathbb{E}\!\left[e^{-\rho\alpha_C
    \min(\ell^i,r^{(i)})(2q^C-\min(\ell^i,r^{(i)}))}\right] 
    + \rho k^i - 1\Bigr).
\end{align*}
For the dark pool allocation, the corresponding Lagrangian is
\begin{align*}
    L(\ell,\vartheta) 
    &= \sum_{i=1}^M \frac{\lambda_0^{(i)}}{\rho}
    \Bigl(-e^{-\rho k^i}\,\mathbb{E}\!\left[e^{-\rho\alpha_C
    \min(\ell^i,r^{(i)})(2q^C-\min(\ell^i,r^{(i)}))}\right]
    + \rho k^i - 1\Bigr) 
    - \vartheta\Bigl(\sum_{i=1}^M \ell^i - q^C\Bigr).
\end{align*}
By the first-order optimality conditions,
\begin{equation*}
    \nabla_\ell L(\ell,\vartheta) =
    \begin{pmatrix}
        2\lambda_0^{(1)}\rho\alpha_C(q^C - \ell^1)
        \exp\!\bigl(-\rho\alpha_C\ell^1(2q^C-\ell^1)\bigr)
        \bigl(1 - F_1(\ell^1)\bigr) - \vartheta \\
        \vdots \\
        2\lambda_0^{(M)}\rho\alpha_C(q^C - \ell^M)
        \exp\!\bigl(-\rho\alpha_C\ell^M(2q^C-\ell^M)\bigr)
        \bigl(1 - F_M(\ell^M)\bigr) - \vartheta
    \end{pmatrix} = 0,
\end{equation*}
from which the optimal dark pool trading strategies follow directly. 
When $M=1$, the gradient is strictly positive for all $\ell^1 < q^C$, 
so the constraint is binding and $\hat\ell_t = Q^C_t$.

For the lit-market rate $\nu$, the objective is strictly concave since $-2\lambda < 0$. Applying the first-order 
condition with the absolute value term gives us the optimal strategy. For the hedging 
rate $h$, the structure is identical with $\alpha_C q^C$ replaced by 
$\alpha_B q^B$. For the internalization rate $u \geq 0$, 
the objective is strictly concave, 
and the maximizer is the unconstrained optimum projected onto $\mathbb{R}_+$.

Finally, for any $\xi\in\mathcal{C}$ satisfying Assumption~\ref{assump:bmo_martingale}, 
by Lemma~\ref{lemma:bsde_y_existence} and Corollary~\ref{corol:unique_y} there 
exists a unique solution $(Y,Z,K)\in\mathcal{S}^2_T\times\mathcal{H}^2_T
\times\mathcal{K}^2_T$ to BSDE \eqref{eq:exp-bsde} with terminal condition $\xi$. 
Evaluating at $t=0$ and rearranging gives the representation 
$\xi = \xi^{Y_0,Z,K} \in \Xi$, hence $\mathcal{C} = \Xi$. Under this 
representation, the value function is
\[
    V^B_0 = -e^{-\rho\bar{Y}_0}
    = -e^{-\rho(Y_0 + X^B_0 + Q^B_0(S_0 - \alpha_B Q^B_0) - \alpha_C(Q^C_0)^2)},
\]
which completes the proof.
\end{proof}

\subsection{Investor's Optimization Problem}\label{sec:investor}

We now solve the upper level of the contracting problem, in which the investor designs the compensation $\xi$ paid to the broker. The resulting leader-follower equilibrium fits the principal-agent framework of \cite{cvitanic2018dynamic,euch2021optimal,baldacci2023market} and we thus reduce the study to a verification theorem for an integro-partial HJB equation. Let $\Psi^C:\R\to\R$ be the investor's utility, assumed increasing and concave, and assume the broker liquidates any residual inventory $Q^C_T$ at the terminal mid-price $S_T$. The investor's terminal P\&L is then
\begin{equation}\label{eq:investor-pnl}
    \Pi^C_T := X^C_T + S_T Q^C_T - \int_0^T \kappa^C_{\mathrm{risk}}(Q^C_t)^2\,dt - \xi,
\end{equation}
where the running quadratic term with parameter $\kappa^C_{\mathrm{risk}}>0$ captures the investor's aversion to carrying inventory during execution. The investor maximizes the expected utility of $\Pi^C_T$ subject to the broker's participation, yielding the bilevel program
\begin{align}\label{eq:investor-bilevel}
    J_0^C &= \sup_{\xi\in\mathcal C}\,\E\!\left[\Psi^C(\Pi^C_T)\right]\\
    \nonumber
    &\text{subject to}\quad V_0^B(\xi)\ge R_0,
\end{align}
with $R_0\in(-\infty,0)$ a fixed reservation utility and the broker playing the closed-loop response $(\hat\nu,\hat h,\hat u,\hat\ell)$ from Theorem~\ref{thm:bsde_value}. We first turn \eqref{eq:investor-bilevel} into a Markovian control problem. By Theorem~\ref{thm:bsde_value}, every $\xi\in\mathcal C$ admits a unique representation $\xi = \xi^{Y_0,Z,K}$ with $(Z,K)\in\H_T^2(\R)\times\K_T^2(\R^M)$, and the broker's value is $V_0^B = -e^{-\rho\bar Y_0}$ for $\bar Y$ defined in \eqref{eq:Ybar_def}. The participation constraint then becomes
\[
    \bar Y_0\ge -\rho^{-1}\log(-R_0),
\]
and the investor's problem reduces to a stochastic control problem over $(Y_0,Z,K)$. The broker's response is Markovian in $(Q^C,Q^B,Z)$. We work under $\P^0$ with the broker's best response substituted in, using the dynamics of $(Q^C,Q^B,I,S,X^C)$ from Section~\ref{sec:market} and \eqref{eq:exp-bsde} for $Y$. To absorb the running cost, we introduce the auxiliary accumulator
\[
    A_t := \int_0^t \kappa^C_{\mathrm{risk}}(Q^C_s)^2\,ds, \qquad A_0 = 0,
\]
so that $\Pi^C_T = X^C_T + S_T Q^C_T - A_T - Y_T$. The augmented state is then $\mathbf x := (q^C,q^B,I,s,x^C,y,a)\in\R^7$, where $s$ denotes the mid-price coordinate; the half-spread $\delta$ of Section~\ref{sec:market} enters the operator below only implicitly, through the broker's closed-loop formulas of Theorem~\ref{thm:bsde_value}. We set
\[
    w(t,\mathbf x) := \sup_{Z,K}\,\E^{\P^0}\!\Bigl[\Psi^C\!\bigl(X^C_T+S_T Q^C_T-A_T-Y_T\bigr)\,\Big|\,(Q^C_t,Q^B_t,I_t,S_t,X^C_t,Y_t,A_t) = \mathbf x\Bigr];
\]
dynamic programming then gives the HJB equation
\begin{equation}\label{eq:investor-HJB}
    \begin{cases}
    \partial_t w + \sup_{z\in\R,\,k\in\R^M}\mathcal G^{z,k}[w] = 0,& t<T,\\
    w(T,\mathbf x) = \Psi^C\!\bigl(x^C + s q^C - a - y\bigr),
    \end{cases}
\end{equation}
with infinitesimal operator
\begin{align*}
    \mathcal G^{z,k}[w](t,\mathbf x)
    &= (\hat\nu_t-\hat u_t)\partial_{q^C}w + (\hat u_t+\hat h_t)\partial_{q^B}w + \bigl(-\kappa I+\eta(\hat\nu_t+\hat h_t+\varsigma_t)\bigr)\partial_I w\\
    &\quad + \bigl((\gamma+\eta)(\hat\nu_t+\hat h_t)+(\epsilon+\eta)\varsigma_t-\kappa I\bigr)\partial_s w + \tfrac12\sigma^2\partial_{ss}w\\
    &\quad + \bigl((s+\varepsilon^{\mathrm{int}})\hat u_t - s\hat\nu_t\bigr)\partial_{x^C}w + \kappa^C_{\mathrm{risk}}(q^C)^2\partial_a w\\
    &\quad - H^B(t,z,k)\partial_y w + \tfrac12 z^2\partial_{yy}w + \sigma z\,\partial_{sy}w\\
    &\quad + \sum_{i=1}^M \lambda_0^{(i)}\,\E_{r^{(i)}}\!\Bigl[w\bigl(t,q^C-\min(\hat\ell^i,r^{(i)}),q^B,I,s,x^C+(s+\varepsilon^{(i)})\min(\hat\ell^i,r^{(i)}),y+k^i,a\bigr) - w(t,\mathbf x)\Bigr].
\end{align*}
Here $(\hat\nu_t,\hat h_t,\hat u_t,\hat\ell_t^i)$ are evaluated at $(t,q^C,q^B,z)$ via Theorem~\ref{thm:bsde_value}.

\begin{theorem}[Verification theorem for the investor]\label{thm:verif_investor}
Assume that there exists
$w\in\mathcal C^{1,1,1,2,1,1,2,1}([0,T]\times\R^7)$ solving the HJB
\eqref{eq:investor-HJB}, with polynomial growth in $(I,s,x^C,y,a)$ and
quadratic growth in $(q^C,q^B)$:
\[
    |w(t,\mathbf x)| \le C\bigl(1+|I|^p+|s|^p+|x^C|^p+|y|^p+|a|^p+|q^C|^2+|q^B|^2\bigr),
    \quad p>1,\;C>0.
\]
Let $(\hat z,\hat k):[0,T]\times\R^7\to\R\times\R^M$ be Borel-measurable
maximizers of $\mathcal G^{z,k}[w]$ such that the controlled SDE for
$\mathbf X_t = (Q^C_t,Q^B_t,I_t,S_t,X^C_t,Y_t,A_t)$ under the feedback
$\hat Z_t=\hat z(t,\mathbf X_t),\,\hat K_t=\hat k(t,\mathbf X_t)$ admits a
unique strong solution. If moreover the resulting initial value
$\bar Y_0$ satisfies the participation constraint
$\bar Y_0\ge -\rho^{-1}\log(-R_0)$, then $(\hat Z,\hat K)$ is optimal for
\eqref{eq:investor-bilevel} and
\[
    J_0^C = w(0,Q_0^C,Q_0^B,I_0,S_0,X_0^C,Y_0,0).
\]
\end{theorem}

\begin{proof}
We follow the standard verification argument (see, e.g., \cite{oksendal2007applied}). Write $\mathbf x_0 := (Q_0^C,Q_0^B,I_0,S_0,X_0^C,Y_0,0)$ and define the localization times
\[
    \tau_n := \inf\bigl\{t\in[0,T]:\|\mathbf X_t\|_2 > n\bigr\}\wedge T.
\]

Fix an admissible $(Z,K)$. Itô's formula applied to $w(t,\mathbf X_t)$ along the controlled dynamics, followed by expectation, gives
\begin{align*}
   w(0,\mathbf x_0)
   &= \E\bigl[w(\tau_n,\mathbf X_{\tau_n})\bigr]
   - \E\!\left[\int_0^{\tau_n}\!\!\bigl(\partial_t w + \mathcal G^{Z_s,K_s}[w]\bigr)(s,\mathbf X_s)\,ds\right],
\end{align*}
since the stochastic integrals up to $\tau_n$ are true martingales.
Since $w$ solves \eqref{eq:investor-HJB}, the integrand is non-positive, so
\[
    w(0,\mathbf x_0)\ge \E\bigl[w(\tau_n,\mathbf X_{\tau_n})\bigr],
\]
with equality when $(Z_t,K_t) = (\hat z(t,\mathbf X_t),\hat k(t,\mathbf X_t))$.

We now pass to the limit $n\to\infty$. Under any admissible control, $\xi\in\mathcal C$ and Assumption~\ref{assump:bmo_martingale} give
\[
    \E\Bigl[\sup_{0\le t\le T}\!\bigl(|S_t|^p+|X_t^C|^p+|Y_t|^p+|I_t|^p+|A_t|^p+|Q_t^C|^2+|Q_t^B|^2\bigr)\Bigr]<\infty.
\]
Combined with the growth assumption on $w$, dominated convergence yields
\[
    w(0,\mathbf x_0) \ge \E\bigl[\Psi^C(X_T^C+S_T Q_T^C-A_T-Y_T)\bigr] = \E\bigl[\Psi^C(\Pi^C_T)\bigr],
\]
with equality under the feedback $(\hat z,\hat k)$. This feedback satisfies the participation constraint by hypothesis, so $(\hat Z,\hat K)$ is admissible for \eqref{eq:investor-bilevel}. Taking the supremum gives $J_0^C = w(0,\mathbf x_0)$.
\end{proof}

\begin{remark}[Linear utility]\label{rem:linear_utility_inv}
Under $\Psi^C(x) = x$, the terminal condition $w(T,\mathbf x) = x^C + sq^C - a - y$ is affine in $(x^C, sq^C, a, y)$. This suggests the ansatz
\[
    w(t,\mathbf x) = x^C + s q^C - a - y + \widetilde w(t,q^C,q^B,I), \qquad \widetilde w(T,\cdot) = 0.
\]
Plugging it into the operator, the explicit $s$-dependence cancels: the term $s(\hat\nu - \hat u)$ from $\partial_{q^C}w$ is exactly offset by $s\hat u - s\hat\nu$ from $\partial_{x^C}w$. The reduced equation for $\widetilde w$ is then a four-dimensional HJB in $(t,q^C,q^B,I)$, with a jump kernel involving only $\varepsilon^{(i)}\min(\hat\ell^i,r^{(i)}) - k^i$.

The investor's objective is decreasing in $Y_0$ through the $-Y_T$ term. The participation constraint $\bar Y_0\ge -\rho^{-1}\log(-R_0)$ therefore binds at the optimum:
\[
    \bar Y_0^\star = -\rho^{-1}\log(-R_0).
\]
As in any moral-hazard problem with a risk-neutral principal, the investor extracts the full surplus above the broker's reservation level.
\end{remark}

\begin{remark}[Markovian structure]\label{rem:closedloop_inv}
The broker's best-response controls $(\hat\nu,\hat h,\hat u,\hat\ell)$ are
closed-loop and Markovian in $(Q^C,Q^B,Z)$. Consequently, the HJB
\eqref{eq:investor-HJB} is well-posed once $\hat z,\hat k$ are determined as
feedback in the augmented state, and the system $(\mathbf X,Z,K)$ is jointly
Markovian.
\end{remark}

\section{Direct execution against strategic high-frequency traders}\label{sec:mfg_hft}

In Section~\ref{sec:investor}, the investor delegates execution to the broker and pays $\xi$ for information privilege, internalization, and hedging. We now turn to a purely competitive model with direct operations on open markets. Without the broker, the HFTs are no longer passive: they observe the investor's lit footprint and trade against it strategically. This is a major-minor mean-field game with the investor as the major player and the HFTs as a strategic minor continuum, following \cite{carmona_lacker_2015,firoozi2017execution,chen2024periodic}. The remainder of the section solves this game; the resulting major value $V^0$ then serves as the no-broker benchmark in Remark~\ref{rem:value_of_broker}.

\subsection{Market dynamics under direct execution}

The broker's contract-side controls are absent in this section: $u_t \equiv h_t \equiv 0$. To focus on the strategic impact channel, we also set $\delta = 0$ and $\varepsilon^{(i)} = 0$.

The investor controls the lit rate $\nu_t^0\in\R$ and the dark-pool allocations $\ell_t^i\ge 0$ with $\sum_i\ell_t^i\le Q_t^0$, where $Q_t^0$ is their remaining inventory ($Q_0^0 = Q_0^C$). A representative HFT trades at rate $\varsigma_t\in\R$, holds inventory $Q_t^H$, and carries cash $X_t^H$. Their cross-sectional inventory distribution is $m(t,dq)$, and the mean-field rate is
\[
    \mu_t := \int_\R \varsigma(t,q)\,m(t,dq).
\]

The mid-price and propagator now interact with $\mu_t$ instead of an exogenous HFT rate:
\begin{equation}\label{eq:mfg-price}
    dS_t = \bigl((\gamma+\eta)\nu_t^0 + (\epsilon+\eta)\mu_t - \kappa I_t\bigr)\,dt + \sigma\,dW_t,\qquad
    dI_t = \bigl(-\kappa I_t + \eta(\nu_t^0+\mu_t)\bigr)\,dt.
\end{equation}
The investor's inventory and cash evolve as
\[
    dQ_t^0 = \nu_t^0\,dt - \sum_{i=1}^M\min(\ell_t^i,r^{(i)})\,dN_t^{(i)},\qquad
    dX_t^0 = -(S_t+\lambda\nu_t^0)\nu_t^0\,dt + \sum_{i=1}^M S_t\,\min(\ell_t^i,r^{(i)})\,dN_t^{(i)},
\]
and a representative HFT obeys $dQ_t^H = \varsigma_t\,dt$ and $dX_t^H = -(S_t+\lambda\varsigma_t)\varsigma_t\,dt$.

\subsection{Strategic HFTs: minor-player mean-field game}\label{subsec:mfg_hft_minor}

We fix the investor's lit strategy as a feedback $\nu^0(t,\bar Q^0_t)$ in their expected inventory $\bar Q^0_t := \E[Q_t^0]$, which the HFT crowd reads off from the time-averaged order flow. The mid-price then becomes
\[
    dS_t = \bigl((\gamma+\eta)\nu^0(t,\bar Q^0_t) + (\epsilon+\eta)\mu_t - \kappa I_t\bigr)\,dt + \sigma\,dW_t.
\]

A representative HFT observes only the investor's lit trajectory and solves
\begin{equation}\label{eq:mfg_hft_value}
    \begin{cases}
        V_0^H(\mu;\nu^0) := \sup_{\varsigma} \E\bigl[X_T^H+Q_T^H(S_T-\alpha_H Q_T^H)\bigr],\\
        \mu_t = \int_\R \hat\varsigma(t,q)\,m(t,dq),\quad t<T,
    \end{cases}
\end{equation}
where $\alpha_H>0$ is their terminal inventory penalty. Since $\mu$ enters \eqref{eq:mfg_hft_value} implicitly through $S$ and $I$, the problem closes as a fixed point on $\mu$.

By dynamic programming, the value function reduces to a scalar $h(t,q)$ solving
\begin{equation}\label{eq:hft-hjb}
    \begin{cases}
    \partial_t h(t,q)
    + \sup_{\varsigma}\Bigl\{-\lambda\varsigma^2
    + \partial_q h(t,q)\,\varsigma
    + \bigl((\gamma+\eta)\nu^0(t,\bar Q^0_t)+(\epsilon+\eta)\mu_t-\kappa I_t\bigr)\,q\Bigr\} = 0,\\
    h(T,q) = -\alpha_H q^2.
    \end{cases}
\end{equation}

\begin{lemma}\label{lem:hft_optimal}
Fix $\mu$, $I$ and the feedback $\nu^0$, and assume \eqref{eq:hft-hjb} admits a unique solution $h^\mu\in\mathcal C^{1,1}([0,T]\times\R)$. The HFT's optimal trading rate is then
\[
    \hat\varsigma_t = \frac{1}{2\lambda}\partial_q h^\mu(t,Q_t^H).
\]
\end{lemma}

\begin{proof}[Sketch]
Plug the ansatz $H(t,q,x,s) = x + sq + h(t,q)$ into the dynamic value, run dynamic programming to obtain \eqref{eq:hft-hjb}, and take the first-order condition in $\varsigma$. The verification step is standard.
\end{proof}

Substituting $\hat\varsigma$ into \eqref{eq:hft-hjb} and coupling it with the Fokker-Planck equation for $m$ and the ODEs for $\bar Q^0$ and $I$ yields the system
\begin{equation}\label{eq:mfg-FK}
    \begin{cases}
    -\bigl((\gamma+\eta)\nu^0(t,\bar Q^0_t)+(\epsilon+\eta)\mu_t-\kappa I_t\bigr)\,q
    = \partial_t h(t,q) + \dfrac{1}{4\lambda}\bigl(\partial_q h(t,q)\bigr)^2,\\
    h(T,q) = -\alpha_H q^2,\\
    \partial_t m(t,q) + \partial_q\!\Bigl(m(t,q)\dfrac{\partial_q h(t,q)}{2\lambda}\Bigr) = 0,\\
    m(0,q) = m_0(q),\\
    \mu_t = \displaystyle\int \frac{\partial_q h(t,q)}{2\lambda}\,m(t,dq),\\
    d\bar Q_t^0 = \Bigl(\nu^0(t,\bar Q^0_t)-\displaystyle\sum_{i=1}^M\lambda_0^{(i)}\E[\min(\ell_t^i,r^{(i)})]\Bigr)\,dt,\\
    dI_t = (-\kappa I_t + \eta(\nu^0(t,\bar Q^0_t)+\mu_t))\,dt.
    \end{cases}
\end{equation}

\begin{definition}[HFT mean-field equilibrium]\label{def:hft_equilibrium}
A triple $(h,m,\mu)$ is an HFT mean-field equilibrium on $[0,T]$ associated
with $\nu^0$ if
\begin{enumerate}
    \item $h$ solves the HJB \eqref{eq:hft-hjb};
    \item $m(t,q)$ solves the Fokker-Planck equation driven by
          $\hat\varsigma_t=\frac{1}{2\lambda}\partial_q h(t,q)$;
    \item $\mu$ satisfies the consistency condition $\mu_t=\int_\R\hat\varsigma(t,q)\,m(t,dq)$.
\end{enumerate}
\end{definition}

\begin{remark}\label{rem:E_mu_identity}
Let $E(t):=\int q\,m(t,dq)$ denote the expected HFT inventory. Then
$E'(t) = \mu_t$, by the standard computation
\[
   E'(t) = -\int q\,\partial_q\!\Bigl(m(t,q)\tfrac{\partial_q h(t,q)}{2\lambda}\Bigr)dq
        = \int \tfrac{\partial_q h(t,q)}{2\lambda}\,m(t,dq) = \mu_t.
\]
\end{remark}

Following \cite{cardaliaguet2018mean,fu2021mean}, we look for a quadratic solution to \eqref{eq:hft-hjb},
\[
    h(t,q) = h_0(t) + h_1(t) q + h_2(t) q^2.
\]
Matching powers of $q$ collapses the HJB into a coupled ODE system,
\begin{equation}\label{eq:hft-ode}
    \begin{cases}
        h_0'(t) + \dfrac{1}{4\lambda}h_1(t)^2 = 0,\\[2pt]
        h_1'(t) + \dfrac{1}{\lambda}h_1(t)h_2(t)
        + (\gamma+\eta)\nu^0(t,\bar Q_t^0)+(\epsilon+\eta)\mu_t-\kappa I_t = 0,\\[2pt]
        h_2'(t) + \dfrac{1}{\lambda}h_2(t)^2 = 0,
    \end{cases}
\end{equation}
with terminal conditions $h_0(T) = h_1(T) = 0$ and $h_2(T) = -\alpha_H$. The optimal HFT rate becomes affine in $q$:
\[
    \hat\varsigma_t = \frac{h_1(t)+2h_2(t)Q_t^H}{2\lambda} = A(t) + B(t)Q_t^H,\quad A(t):=\frac{h_1(t)}{2\lambda},\;B(t):=\frac{h_2(t)}{\lambda}.
\]
By Remark~\ref{rem:E_mu_identity}, the mean field and the expected inventory $E(t)$ are tied by
\begin{equation}\label{eq:mfg-mu-E}
    \mu_t = \frac{h_1(t)}{2\lambda} + \frac{h_2(t)}{\lambda}E(t), \qquad
    2\lambda E'(t) = h_1(t) + 2h_2(t)E(t).
\end{equation}

Differentiating the second identity in \eqref{eq:mfg-mu-E} and substituting the ODEs for $h_1$ and $h_2$, the cross terms cancel and we are left with a linear second-order ODE for $E$:
\begin{equation}\label{eq:E_linear}
    2\lambda E''(t) + (\epsilon+\eta) E'(t) = -(\gamma+\eta)\nu^0(t,\bar Q_t^0) + \kappa I_t,\quad t\in[0,T].
\end{equation}
Evaluating $2\lambda E'(T) = h_1(T) + 2h_2(T)E(T)$ at the terminal time with $h_1(T) = 0$ and $h_2(T) = -\alpha_H$ produces the Robin condition
\begin{equation}\label{eq:robin_T}
    E'(T) + \frac{\alpha_H}{\lambda} E(T) = 0.
\end{equation}
The initial condition $E(0) = E_0$ comes from $m_0$.

\begin{theorem}[Existence and uniqueness of the HFT MFG equilibrium]\label{thm:mfg_hft_existence}
Let $\nu^0\in L^2([0,T])$ be fixed, let $\ell$ be a
feasible dark-pool allocation, and let $I\in C([0,T])$ be the unique solution
to the transient-impact ODE in \eqref{eq:mfg-FK} driven by
$(\nu^0,\mu)$. Then
\begin{enumerate}
    \item[(i)] The Riccati ODE for $h_2$ admits the global explicit solution
    \[
        h_2(t) = -\frac{\alpha_H\,\lambda}{\lambda + \alpha_H(T-t)},\quad t\in[0,T].
    \]
    \item[(ii)] The linear BVP \eqref{eq:E_linear}-\eqref{eq:robin_T}
    admits a unique solution $E\in\mathcal C^2([0,T])$, and so do $h_1$ and
    $h_0$ from the remaining ODEs in \eqref{eq:hft-ode}.
    \item[(iii)] The map $\Phi:\widetilde\mu\mapsto \int_\R\hat\varsigma(t,q)\,m(t,dq)$
    is a contraction on $L^2([0,T])$; consequently there is a unique HFT
    mean-field equilibrium $(h,m,\mu)$ on $[0,T]$ in the sense of
    Definition~\ref{def:hft_equilibrium}, and the triple constructed from
    \eqref{eq:hft-ode}-\eqref{eq:robin_T} yields a global equilibrium.
\end{enumerate}
\end{theorem}

\begin{proof}
\textbf{Step 1.} (i) follows by direct integration of the Riccati
$h_2'+\lambda^{-1}h_2^2=0$ subject to $h_2(T)=-\alpha_H$; the denominator
$\lambda+\alpha_H(T-t)>0$ on $[0,T]$ ensures the solution is global.

\textbf{Step 2.} The ODE \eqref{eq:E_linear} is linear of second order with continuous coefficients on the compact $[0,T]$. Together with $E(0) = E_0$ and the Robin condition \eqref{eq:robin_T}, it has a unique solution $E\in\mathcal C^2([0,T])$ by standard two-point BVP theory. We then read off $h_1(t) = 2\lambda E'(t) - 2h_2(t)E(t)$, which is $\mathcal C^1$, and recover $h_0$ by integrating $h_0'(t) = -\frac{1}{4\lambda}h_1(t)^2$ backwards from $h_0(T) = 0$.

\textbf{Step 3.} With $\hat\varsigma_t=A(t)+B(t)q$ and $A,B\in L^1([0,T])$,
the continuity equation
\[
    \partial_t m + \partial_q\bigl((A(t)+B(t)q)m\bigr) = 0,\quad m(0,\cdot)=m_0,
\]
is solved by the method of characteristics. Define
\[
    \Gamma(t):=\exp\!\Bigl(\int_0^t B(u)\,du\Bigr),\quad
    \psi(t):=\int_0^t \Gamma(t)\Gamma(s)^{-1}A(s)\,ds.
\]
The unique measure solution is the pushforward
$m(t,\cdot)=(\Gamma_{t,0})_\#m_0$ with $\Gamma_{t,0}(q)=\Gamma(t)q+\psi(t)$, and
$E(t)=\Gamma(t)E_0+\psi(t)$.

\textbf{Step 4.} Consider the map $\Phi:\widetilde\mu\mapsto A(t) + B(t)E(t)$. The triple $(h_1,h_2,E)$ depends Lipschitz-continuously on $\widetilde\mu$ through the linear BVP \eqref{eq:E_linear}-\eqref{eq:robin_T}, and $m$ depends affinely on $(A,B)$ through the pushforward of Step 3. Combining the two, \cite[Theorem~5.1]{cardaliaguet2018mean} gives a uniform contraction modulus for $\Phi$ on $L^2([0,T])$.

Steps 1-4 together show that $(h,m,\mu)$ exists and is unique, and that the equilibrium is the one constructed from the explicit ODE solution.
\end{proof}

\begin{remark}[Comparison with the exogenous-HFT model]
In Section~\ref{sec:investor}, $\varsigma_t$ was an exogenous adapted process. Here it is endogenous: $\hat\varsigma_t = A(t) + B(t)Q_t^H$, with $A$ and $B$ depending on the investor's strategy $\nu^0$ through the linear BVP \eqref{eq:E_linear}. The equilibrium mean field $\mu_t$ therefore tracks $\nu^0$ through a delayed linear response, in line with the front-running behavior of HFTs around large institutional orders documented in \cite{van2019high} and modeled in \cite{baldacci2023mean,cartea2024nash}.
\end{remark}

\subsection{Investor's problem against the strategic HFT mean field}\label{subsec:mfg_major}

We now turn to the major player. Fix the equilibrium triple $(h,m,\mu)$ from Theorem~\ref{thm:mfg_hft_existence}. The investor liquidates $Q_0^0 = Q_0^C$ over $[0,T]$ across one lit venue and $M$ dark pools, maximizing
\[
    V^0(Q_0^0;\mu) := \sup_{\nu^0,\ell}\E\!\left[X_T^0+Q_T^0(S_T-\alpha_C Q_T^0) - \int_0^T \kappa^C_{\mathrm{risk}}(Q_t^0)^2\,dt\right].
\]
Here $\alpha_C > 0$ is the same terminal inventory penalty as in the contracting regime, but now borne by the investor since no broker absorbs the residual position. The running penalty $\kappa^C_{\mathrm{risk}}$ also matches the one in \eqref{eq:investor-pnl}.

Setting $H^0(t,x,s,q^0) = x + sq^0 + h^0(t,q^0)$, dynamic programming gives the reduced HJB
\begin{equation}\label{eq:major-hjb-mfg}
\begin{cases}
    \partial_t h^0(t,q^0)
    + \sup_{\nu^0}\Bigl\{-\lambda(\nu^0)^2 - \kappa^C_{\mathrm{risk}}(q^0)^2
    + (\partial_{q^0}h^0+(\gamma+\eta)q^0)\nu^0 + ((\epsilon+\eta)\mu_t-\kappa I_t)q^0\Bigr\}\\
    \qquad + \displaystyle\sup_{\ell}\sum_{i=1}^M \lambda_0^{(i)}\,\E_r\!\bigl[h^0(t,q^0-\min(\ell^i,r^{(i)}))-h^0(t,q^0)\bigr] = 0,\\
    h^0(T,q^0) = -\alpha_C(q^0)^2.
\end{cases}
\end{equation}

\begin{remark}[Concavity]\label{rem:major-concave}
The solution $h^0(t,q^0)$ to \eqref{eq:major-hjb-mfg} admits a Feynman-Kac representation through a BSDE with terminal $-\alpha_C(Q_T^0)^2$ and generator
\[
    f^{\nu^0}(z,q^0) = -\lambda(\nu^0)^2 - \kappa^C_{\mathrm{risk}}(q^0)^2 + (z+(\gamma+\eta)q^0)\nu^0 + ((\epsilon+\eta)\mu_t-\kappa I_t)q^0.
\]
Both the terminal and the generator are concave in $q^0$, so $h^0$ inherits this concavity by Lemma~\ref{lemma:comparison}.
\end{remark}

\begin{theorem}[Investor's optimal strategy against strategic HFTs]\label{thm:investor_major}
Under Remark~\ref{rem:major-concave}, the investor's optimal controls are characterized by
\begin{enumerate}
    \item Lit-market trading rate:
    \[
        \hat\nu_t^0
        = \left(\frac{\partial_{q^0}h^0(t,Q_t^0)+(\gamma+\eta)Q_t^0}{2\lambda}\right)^{\!-},
    \]
    where $x^-:=\min(x,0)$ denotes the negative part.
    \item Dark-pool allocation. Define $b(t):=(\partial_{q^0}h^0(t,\cdot))^{-1}(0)$.
    Then
    \begin{enumerate}
        \item If $M=1$: $\hat\ell_t = \min\bigl(Q_t^0,(Q_t^0-b(t))^+\bigr)$.
        \item If $M>1$: for any $i,j\in\{1,\dots,M\}$ and $t\in[0,T]$,
        \[
            \lambda_0^{(i)}\bigl(1-F_i(\hat\ell^i)\bigr)\partial_{q^0}h^0\!\bigl(t,Q_t^0-\hat\ell^i\bigr)
            = \lambda_0^{(j)}\bigl(1-F_j(\hat\ell^j)\bigr)\partial_{q^0}h^0\!\bigl(t,Q_t^0-\hat\ell^j\bigr).
        \]
    \end{enumerate}
\end{enumerate}
\end{theorem}

\begin{proof}[Sketch of proof]
The lit-market FOC follows by direct differentiation of the strictly concave
$\nu^0\mapsto -\lambda(\nu^0)^2 + (\partial_{q^0}h^0+(\gamma+\eta)q^0)\nu^0$,
restricted to $\nu^0\le 0$ since the investor is a seller. For the dark-pool
allocation, the Lagrangian
\[
    L(\ell,\vartheta) = \sum_{i=1}^M\lambda_0^{(i)}\E_r\bigl[h^0(t,q^0-\min(\ell^i,r^{(i)}))-h^0(t,q^0)\bigr]
    + \vartheta\Bigl(q^0-\sum_i\ell^i\Bigr)
\]
yields the first order condition
$\nabla_{\ell^i}L = -\lambda_0^{(i)}(1-F_i(\ell^i))\partial_{q^0}h^0(t,q^0-\ell^i)-\vartheta = 0$,
which gives the stated equalization condition.
\end{proof}

\begin{remark}[Stackelberg equilibrium]
Theorems~\ref{thm:mfg_hft_existence} and \ref{thm:investor_major} characterize the two levels of the Stackelberg game. Given a committed lit strategy $\nu^0$, Theorem~\ref{thm:mfg_hft_existence} returns the HFT response $(h,m,\mu)$; given the resulting $\mu$, Theorem~\ref{thm:investor_major} returns the investor's best response.

The equilibrium is the fixed point of the alternating map $\nu^0 \mapsto \mu(\nu^0)\mapsto \nu^0(\mu)$. Under the quadratic ansatz it reduces to a coupled ODE system in $(h_\cdot, h_\cdot^0, E, \bar Q^0, I)$, which can be solved numerically by a forward-backward sweep.
\end{remark}

\begin{remark}[The value of broker delegation]\label{rem:value_of_broker}
The two regimes give a direct measure of the broker's economic value. Write $J_0^C$ for the investor's optimal utility under the contracting regime of Section~\ref{sec:investor}, and $V_0^{\mathrm{direct}} := V^0(Q_0^C;\mu^\star)$ for the optimal utility under the direct-execution equilibrium of this section. The difference
\[
    \Delta := J_0^C - V_0^{\mathrm{direct}}
\]
is the maximum surplus the investor would transfer through $\xi$ to obtain the broker's bundled services (information privilege, internalization, and hedging) while keeping their own welfare unchanged. Delegation is welfare-improving when $\Delta > 0$; otherwise the investor prefers to face the strategic HFT crowd directly.
\end{remark}

\section{Numerical experiments}\label{sec:numerics}

We illustrate the two regimes side by side. Section~\ref{sec:num-contract} solves the contract regime via an actor-critic deep-learning method on the reduced HJB for the investor's value $\widetilde w(t,q^C,q^B,I)$. Section~\ref{sec:num-mfg} solves the direct-execution Stackelberg MFG by a forward-backward sweep on the quadratic-ansatz ODE system of Theorem~\ref{thm:mfg_hft_existence}. Section~\ref{sec:num-illustration} feeds the trained policy and the equilibrium feedback into a Monte Carlo simulator with $5{,}000$ sample paths in each regime and compares the resulting trajectories, distributions, and welfare.

\subsection{Contract regime: actor-critic for the investor's HJB}\label{sec:num-contract}

\paragraph*{Reduction.}
Under $\Psi^C(x) = x$ the linear-utility ansatz of Remark~\ref{rem:linear_utility_inv},
\[
    w(t,q^C,q^B,I,s,x^C,y,a) = x^C + s q^C - a - y + \widetilde w(t,q^C,q^B,I),
\]
collapses the seven-dimensional HJB \eqref{eq:investor-HJB} to a four-dimensional PDE in $(t,q^C,q^B,I)$ with terminal condition $\widetilde w(T,\cdot) = 0$. The investor's controls $(Z,K)$ enter the reduced operator only through the broker's closed-loop response $(\hat\nu,\hat h,\hat u,\hat\ell)$ of Theorem~\ref{thm:bsde_value} and the Hamiltonian $H^B(t,Z,K)$ of \eqref{eq:exp-driver}.

\paragraph*{Non-strategic HFTs.}
Section~\ref{sec:investor} carries an exogenous HFT trading rate $\varsigma_t$ that enters the broker's Hamiltonian and the investor's HJB. In the contract-regime experiments below, however, HFTs are by assumption non-strategic: they have no objective and no consistent mechanism to commit to a non-trivial trajectory. We therefore set $\varsigma_t \equiv 0$ in the contract simulator, so that any HFT presence in this regime is restricted to the passive Poisson dark-pool flow with intensity $\lambda_0^{(i)}$. The MFG of Section~\ref{sec:num-mfg} uses the equilibrium mean field $\mu_t$ in place of $\varsigma_t$ and is therefore unaffected.

\paragraph*{Networks and Bellman target.}
Following \cite{baldacci2023market,han2018solving,raissi2024forward,lu2025multiagent,mastrolia2025optimal}, the value function $\widetilde w(t,q^C,q^B,I)$ is approximated by a critic network $\mathcal N_{\widetilde w}$ and the investor's controls $(Z, K)$ by an actor network $\mathcal N_\pi$ with two heads. Both are multi-layer perceptrons with LeakyReLU activations, batch normalization, and orthogonal weight initialization. Time is discretized with step $\Delta t$; the Bellman target at $(t_k, q^C_k, q^B_k, I_k)$ is
\[
    y_k = \mathcal N_{\widetilde w}^{\mathrm{tgt}}\bigl(t_k + \Delta t, q^C_k, q^B_k, I_k\bigr)
    + \mathcal G^{\mathcal N_\pi}[\mathcal N_{\widetilde w}]\bigl(t_k, q^C_k, q^B_k, I_k\bigr)\,\Delta t,
\]
where $\mathcal G^{\mathcal N_\pi}$ is the infinitesimal operator of \eqref{eq:investor-HJB} with $(Z,K)$ replaced by $\mathcal N_\pi$ and the broker's controls evaluated at the maximizer of $H^B$. The expectation against the dark-pool liquidity distribution inside $H^B$ is computed by a $16$-node Gauss-Legendre quadrature on a moving support $[0,\hat\ell^i]$. The critic minimizes the mean-square Bellman residual; the actor maximizes $\mathcal G^{\mathcal N_\pi}[\mathcal N_{\widetilde w}]$. Cosine annealing on both learning rates, with floor at $5\%$ of the initial value, stabilizes the late-stage residual. Algorithm~\ref{alg:contract} summarizes the loop.

\begin{algorithm}[H]
\caption{Actor-Critic update for the contract regime}\label{alg:contract}
\begin{algorithmic}[1]
\STATE \textbf{Input:} Model and training parameters from Table~\ref{tab:params-contract}.
\STATE Initialize critic $\mathcal N_{\widetilde w}$, target critic $\mathcal N_{\widetilde w}^{\mathrm{tgt}}$, and actor $\mathcal N_\pi$.
\FOR{epoch $i \gets 0,1,2,\dots,N_{\mathrm{epoch}}-1$}
    \STATE Uniformly sample a batch $\{(t_k, q^C_k, q^B_k, I_k)\}_{k=1}^{\mathcal K}$ from $[0,T]\times[0,\bar Q^C]\times[-\bar Q^B,\bar Q^B]\times[-\bar I,\bar I]$.
    \STATE Read $(Z_k, K_k) \gets \mathcal N_\pi(t_k, q^C_k, q^B_k, I_k)$ and the broker's closed-loop $(\hat\nu_k, \hat h_k, \hat u_k, \hat\ell_k)$ from Theorem~\ref{thm:bsde_value}.
    \STATE \textbf{Critic step} (every $K_c$ epochs): freeze $\mathcal N_\pi$, compute targets $y_k$, and minimize $\frac{1}{\mathcal K}\sum_k(\mathcal N_{\widetilde w}(t_k,q^C_k,q^B_k,I_k) - y_k)^2$.
    \STATE \textbf{Actor step} (every $K_a$ epochs): freeze $\mathcal N_{\widetilde w}$ and minimize $-\frac{1}{\mathcal K}\sum_k \mathcal G^{\mathcal N_\pi}[\mathcal N_{\widetilde w}](t_k,q^C_k,q^B_k,I_k)$.
    \STATE \textbf{LR step:} apply cosine-annealing update to both learning rates.
    \STATE \textbf{Target update} (every $K_{\mathrm{tg}}$ epochs): $\mathcal N_{\widetilde w}^{\mathrm{tgt}} \leftarrow (1-\tau)\mathcal N_{\widetilde w}^{\mathrm{tgt}} + \tau\,\mathcal N_{\widetilde w}$.
\ENDFOR
\STATE \textbf{Output:} trained networks $\mathcal N_{\widetilde w}^\star, \mathcal N_\pi^\star$.
\end{algorithmic}
\end{algorithm}

\begin{table}[h]
  \centering
  \begin{tabular}{lcl}
  \hline
  Symbol & Value & Comment \\
  \hline
  $T$ & $1$ & trading horizon \\
  $Q_0^C$ & $1$ & investor initial inventory \\
  $Q_0^B$ & $0$ & broker initial inventory \\
  $\lambda$ & $0.02$ & instantaneous lit-market impact \\
  $\eta$ & $0.02$ & transient impact magnitude \\
  $\kappa$ & $1.0$ & transient-impact decay rate \\
  $\gamma,\,\epsilon$ & $0.01$ & permanent impacts (broker / HFT) \\
  $\sigma$ & $0.02$ & mid-price volatility \\
  $\alpha_C,\alpha_B$ & $0.04$ & terminal inventory penalties \\
  $\kappa_{\mathrm{risk}}^C$ & $0.05$ & investor running inventory penalty \\
  $\kappa_{\mathrm{tox}}$ & $0.05$ & internalization toxicity penalty \\
  $\varepsilon^{\mathrm{int}},\varepsilon^{(i)}$ & $0$ & mid-price improvements \\
  $\rho$ & $300$ & broker exponential-utility risk aversion \\
  $R_0$ & $-1.5 \times 10^{-3}$ & broker reservation utility \\
  $\varsigma_t$ & $0$ & exogenous HFT rate (Section~\ref{sec:num-contract}) \\
  $(\lambda_0^{(i)})$ & $(30)$ or $(30,4)$ & dark-pool Poisson intensities \\
  $(a^{(i)})$ & $(200)$ or $(200,10)$ & dark-pool liquidity $r^{(i)}\sim\mathrm{Exp}(a^{(i)})$ \\
  \hline
  $\Delta t$ & $10^{-3}$ & time step \\
  $N_{\mathrm{epoch}}$ & $30{,}000$ & training epochs \\
  $\mathcal K$ & $256$ & batch size \\
  $\eta_{\mathrm{init}}$ & $10^{-3}$ & initial learning rate (both networks) \\
  $\eta_{\mathrm{min}}$ & $5\times 10^{-5}$ & cosine-annealing floor \\
  hidden, depth & $64, 3$ / $64, 2$ & critic / actor MLP \\
  $K_c, K_a, K_{\mathrm{tg}}$ & $1, 2, 50$ & critic/actor/target update periods \\
  $\tau$ & $0.05$ & soft target-update factor \\
  \hline
  \end{tabular}
  \caption{Contract-regime model and training parameters.}
  \label{tab:params-contract}
\end{table}

\subsection{Direct execution: forward-backward sweep for the MFG}\label{sec:num-mfg}

\paragraph*{Reduction.}
The quadratic ansatz $h(t,q) = h_0(t) + h_1(t)q + h_2(t)q^2$ collapses the HJB \eqref{eq:hft-hjb} into the coupled ODE system \eqref{eq:hft-ode}. The Riccati component admits the closed form
\[
    h_2(t) = -\frac{\alpha_H \lambda}{\lambda + \alpha_H(T-t)},
\]
the $q^1$-coefficient $h_1$ is recovered through the identity $h_1 = 2\lambda E' - 2 h_2 E$, and $h_0$ is integrated backward from $h_0(T) = 0$. The expected HFT inventory $E(t) = \int q\,m(t,dq)$ satisfies the linear two-point boundary-value problem
\[
    2\lambda E''(t) + (\epsilon+\eta) E'(t) = -(\gamma+\eta)\nu^0(t,\bar Q^0_t) + \kappa I_t,\qquad E(0) = E_0,\quad E'(T) + \frac{\alpha_H}{\lambda}E(T) = 0,
\]
which we discretize on a uniform grid $\{t_n\}_{n=0}^{N_{\mathrm{time}}}$ with second-order central differences and a one-sided Robin update at $t_N = T$.

\paragraph*{Outer fixed point on $\mu$.}
Equations \eqref{eq:hft-ode}-\eqref{eq:robin_T} are coupled through the consistency $\mu_t = E'(t)$. We solve the joint system by a relaxation iteration: at iterate $i$, given $\mu^{(i)}$, compute $h_2$, then the linear BVP for $E$, then update $\mu^{(i+1)} \leftarrow E'$ and damp by a relaxation factor $\omega \in (0,1)$. The investor's expected inventory $\bar Q^0$ and the OW propagator $I$ are advanced forward in lockstep. The major player's lit feedback $\hat\nu^0(t, \bar Q^0)$ is read off the closed-form formula of Theorem~\ref{thm:investor_major} at each backward sweep. Convergence is declared when $\lVert \mu^{(i+1)} - \mu^{(i)}\rVert_\infty < \varepsilon$. Algorithm~\ref{alg:mfg} summarizes the loop, paralleling Algorithm~2 of the submission.

\begin{algorithm}[H]
\caption{Forward-backward sweep for the direct-execution MFG}\label{alg:mfg}
\begin{algorithmic}[1]
\STATE \textbf{Input:} Model and solver parameters from Table~\ref{tab:params-mfg}; relaxation weight $\omega$, tolerance $\varepsilon$.
\STATE Discretize time $t_n = n\,\Delta t$ on $[0,T]$ and inventory range for $m_0$.
\STATE Initialize $\mu^{(0)} \equiv 0$, $E^{(0)} \equiv E_0$, $\bar Q^{0,(0)}(t) \equiv Q_0^0$, $I^{(0)} \equiv I_0$.
\FOR{$i \gets 0,1,2,\dots$ \textbf{until} convergence}
    \STATE \textbf{(Backward)} Read $h_2(t)$ from the Riccati closed form.
    \STATE \quad Solve the linear BVP \eqref{eq:E_linear}-\eqref{eq:robin_T} for $E^{(i+1)}$ given $(\nu^{0,(i)}, I^{(i)})$.
    \STATE \quad Recover $h_1(t) = 2\lambda E'(t) - 2 h_2(t) E(t)$ and integrate $h_0'(t) = -\frac{1}{4\lambda}h_1(t)^2$ from $h_0(T) = 0$.
    \STATE \quad Read $\hat\nu^{0,(i+1)}(t)$ from Theorem~\ref{thm:investor_major} and the dark-pool allocation $\hat\ell^{(i+1)}$ from the equalization condition.
    \STATE \textbf{(Forward)} March $\bar Q^{0,(i+1)}$ and $I^{(i+1)}$ along $[0,T]$ given $(\hat\nu^{0,(i+1)}, \mu^{(i)})$.
    \STATE \textbf{Relax} $\mu^{(i+1)} \gets (1-\omega)\,\mu^{(i)} + \omega\,E^{(i+1)\prime}$.
    \STATE \textbf{Check} $\lVert \mu^{(i+1)} - \mu^{(i)}\rVert_\infty < \varepsilon$.
\ENDFOR
\STATE \textbf{Output:} equilibrium $(h_0, h_1, h_2)$, $E$, $\mu$, $\bar Q^0$, $I$, and feedback $\hat\nu^0, \hat\ell$.
\end{algorithmic}
\end{algorithm}

\begin{table}[h]
  \centering
  \begin{tabular}{lcl}
  \hline
  Symbol & Value & Comment \\
  \hline
  $T$ & $1$ & trading horizon \\
  $Q_0^0$ & $1$ & investor (major) initial inventory \\
  $\lambda,\eta,\kappa,\gamma,\epsilon,\sigma$ & as in Table~\ref{tab:params-contract} & impact, propagator, volatility \\
  $\alpha_C$ & $0.04$ & investor terminal inventory penalty \\
  $\kappa_{\mathrm{risk}}^C$ & $0.05$ & investor running inventory penalty \\
  $\alpha_H$ & $0.04$ & HFT terminal inventory penalty \\
  $E_0$ & $0.1$ & initial expected HFT inventory \\
  $(\lambda_0^{(i)}), (a^{(i)})$ & as in Table~\ref{tab:params-contract} & dark-pool intensities, liquidity \\
  \hline
  $N_{\mathrm{time}}$ & $1000$ & time grid resolution ($\Delta t = 10^{-3}$) \\
  $\omega$ & $0.5$ & relaxation weight \\
  $\varepsilon$ & $10^{-6}$ & convergence tolerance \\
  $i_{\max}$ & $200$ & maximum outer iterations \\
  \hline
  \end{tabular}
  \caption{Direct-execution MFG model and solver parameters.}
  \label{tab:params-mfg}
\end{table}

\subsection{Illustration and comparison}\label{sec:num-illustration}

Using the trained policy of Section~\ref{sec:num-contract} and the equilibrium feedback of Section~\ref{sec:num-mfg}, we simulate $5{,}000$ Monte Carlo paths in each regime under the parameters of Tables~\ref{tab:params-contract}-\ref{tab:params-mfg}. The contract regime is run for $M=1$ and $M=2$ dark pools; the direct-execution regime is run for the same two configurations and serves as the no-broker benchmark of Remark~\ref{rem:value_of_broker}. In the $M=2$ configuration the two pools are deliberately heterogeneous: pool 1 is a busy crossing venue with frequent matching events and small fills ($\lambda_0^{(1)} = 30$, mean fill $1/a^{(1)} = 0.005$), while pool 2 is a block-trading venue with rare matching events and fills twenty times larger ($\lambda_0^{(2)} = 4$, mean fill $1/a^{(2)} = 0.1$). The block venue is deeper in expectation (per-unit-time capacity $0.4$ versus $0.15$), and the equalization condition of Theorem~\ref{thm:bsde_value} must trade off arrival frequency against fill depth, so the allocation across the two venues is genuinely non-trivial. Throughout the figures, one fixed color per regime is used: blue for the contract regime with $M=1$, green for the contract regime with $M=2$, and amber for direct execution.

\paragraph*{Execution dynamics.}
Figures~\ref{fig:rev_lit_rate}-\ref{fig:rev_dark_volume} describe how the order flow is routed across venues over the horizon. Figure~\ref{fig:rev_lit_rate} reports the lit-market rate at which the client's order is executed, that is, the broker's $\hat\nu_t$ in the contract regime (distinct from the broker's own hedging rate $\hat h_t$ shown in Figure~\ref{fig:rev_broker}) and the investor's own $\hat\nu^0_t$ in the direct regime. Under direct execution the rate is monotone, a consequence of the closed-form Riccati feedback in $Q^0$. The contract regime shows a non-monotone profile with little trading at the beginning, reflecting how the broker uses internalization to smooth the lit footprint at first. Figure~\ref{fig:rev_broker} focuses on the broker's own controls. The internalization rate $\hat u_t$ starts near $0.8$ and decays smoothly to zero by $t \approx 0.8$: the broker absorbs the bulk of the client's position early, when the client's inventory is largest. The hedging rate $\hat h_t$ changes sign over the horizon: it is positive early, while the broker internalizes and also builds inventory on the lit venue, and turns negative later ($t \approx 0.6$ at $M=1$, $t \approx 0.5$ at $M=2$), when the broker offloads its accumulated book. Correspondingly, the broker inventory $\bar Q^B_t$ rises to a peak of about $0.5$ (near $t \approx 0.6$ at $M=1$, slightly earlier at $M=2$) and then unwinds; with the block venue available the unwind is nearly complete ($\bar Q^B_T \approx 0.05$), whereas at $M=1$ the broker retains a residual position of about $0.24$ at the terminal date. Figure~\ref{fig:rev_inventory} translates these dynamics into the client residual position $\bar Q^C_t$: adding the block venue drives $\bar Q^C_T$ from $0.22$ at $M=1$ to $0.03$ at $M=2$, while in the direct regime the strategic-HFT pressure keeps the investor cautious and a sizable residual position remains for both $M$. Figure~\ref{fig:rev_dark_volume} aggregates the dark-pool fills per simulation; the mean total dark execution nearly triples when the block venue is added ($0.15 \to 0.44$ in the contract regime, $0.15 \to 0.39$ under direct execution), confirming that the second venue is a genuine substitute for lit trading. The split across the two heterogeneous pools is itself informative: in the contract regime the block venue absorbs more than twice the volume of the busy crossing venue ($0.31$ versus $0.13$) despite matching seven times less often, because each match executes twenty times more. The equalization condition of Theorem~\ref{thm:bsde_value} trades off arrival frequency against fill depth rather than allocating proportionally to intensity.

\begin{figure}[htbp]
    \centering
    \includegraphics[width=\textwidth]{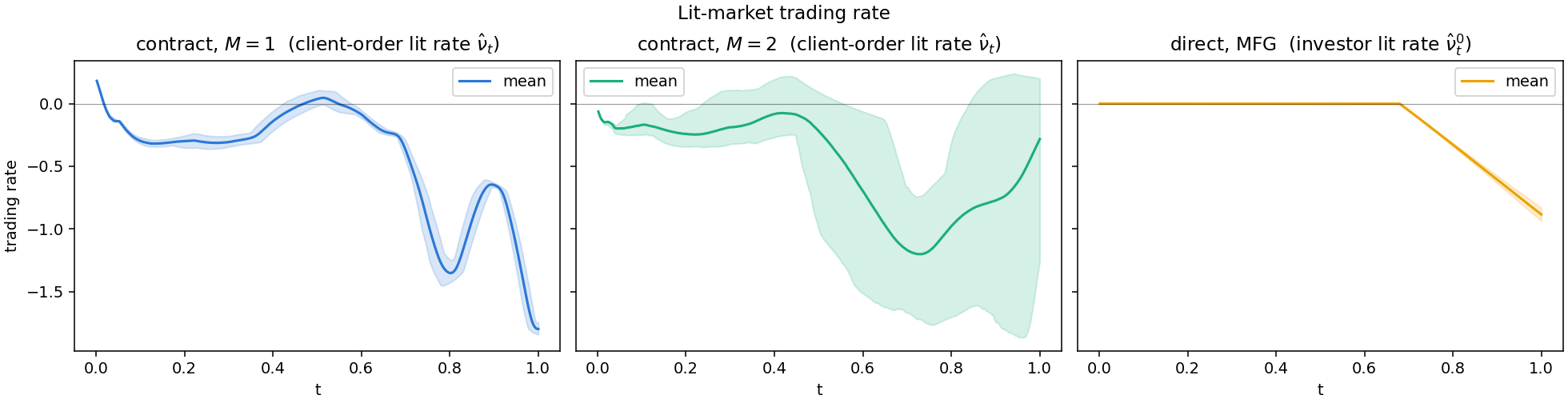}
    \caption{Lit-market trading rate (mean and $80\%$ band) over $5{,}000$ simulations. Left and middle: the broker's lit execution rate on the client's order, $\hat\nu_t$, in the contracting regime with $M=1$ and $M=2$ dark pools, which is distinct from the broker's own hedging rate $\hat h_t$ in Figure~\ref{fig:rev_broker}. Right: the investor's own lit rate $\hat\nu^0_t$ in the direct-execution equilibrium.}
    \label{fig:rev_lit_rate}
\end{figure}

\begin{figure}[htbp]
    \centering
    \includegraphics[width=0.85\textwidth]{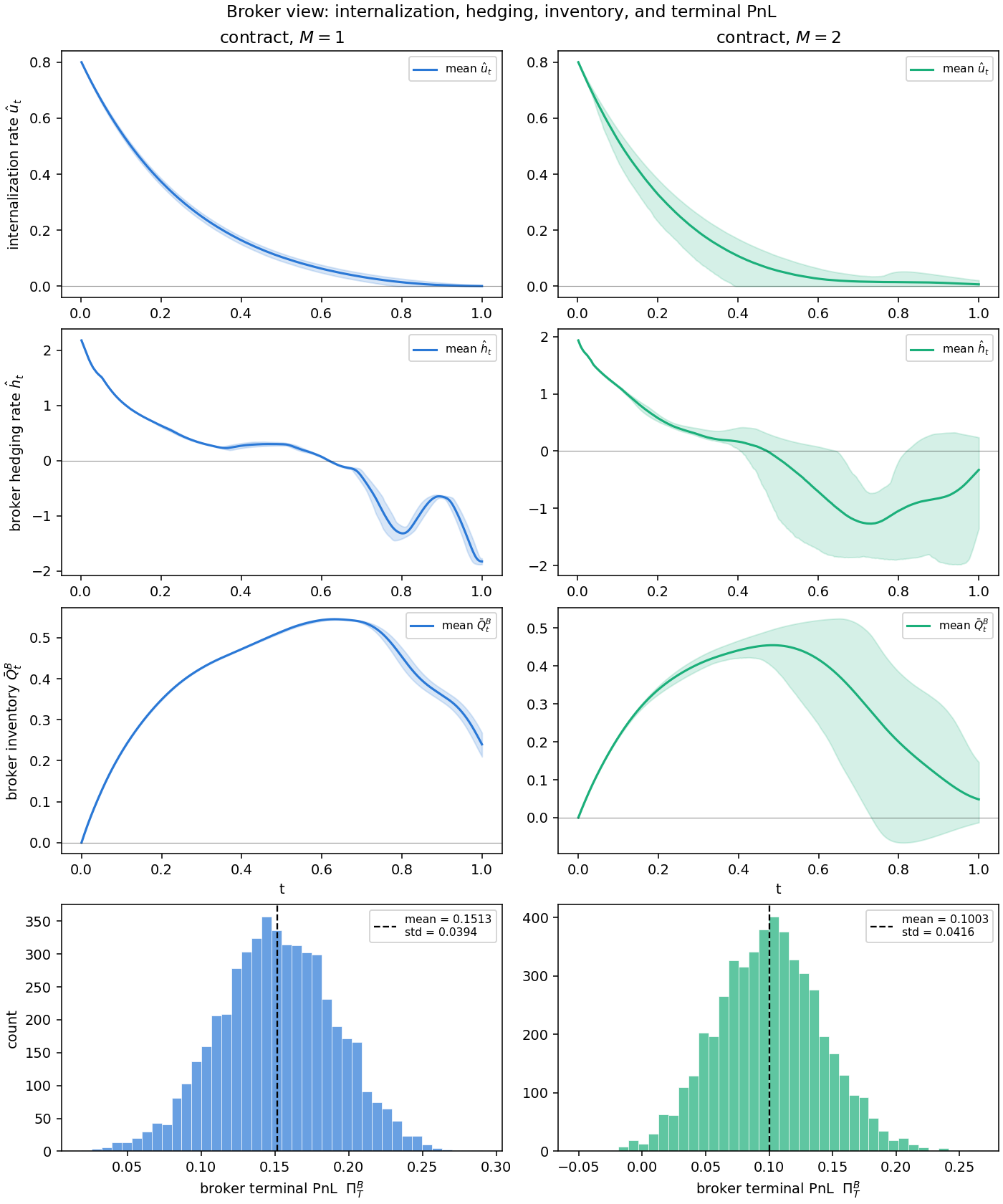}
    \caption{Broker view of the contract regime. Rows from top to bottom: internalization rate $\hat u_t$, hedging rate $\hat h_t$, inventory $\bar Q^B_t$, and terminal broker PnL $\Pi^B_T$ from \eqref{eq:broker-linear-payoff}. Left and right columns are $M=1$ and $M=2$ dark pools respectively.}
    \label{fig:rev_broker}
\end{figure}

\begin{figure}[htbp]
    \centering
    \includegraphics[width=\textwidth]{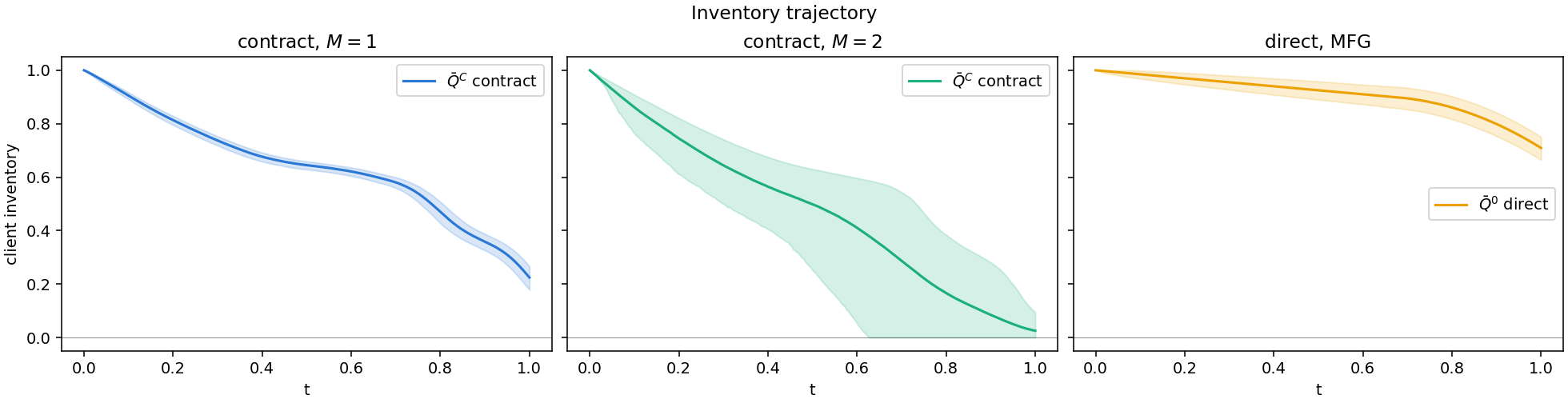}
    \caption{Client inventory $\bar Q^C_t$ (or $\bar Q^0_t$ in the direct regime) over the execution window.}
    \label{fig:rev_inventory}
\end{figure}

\begin{figure}[htbp]
    \centering
    \includegraphics[width=\textwidth]{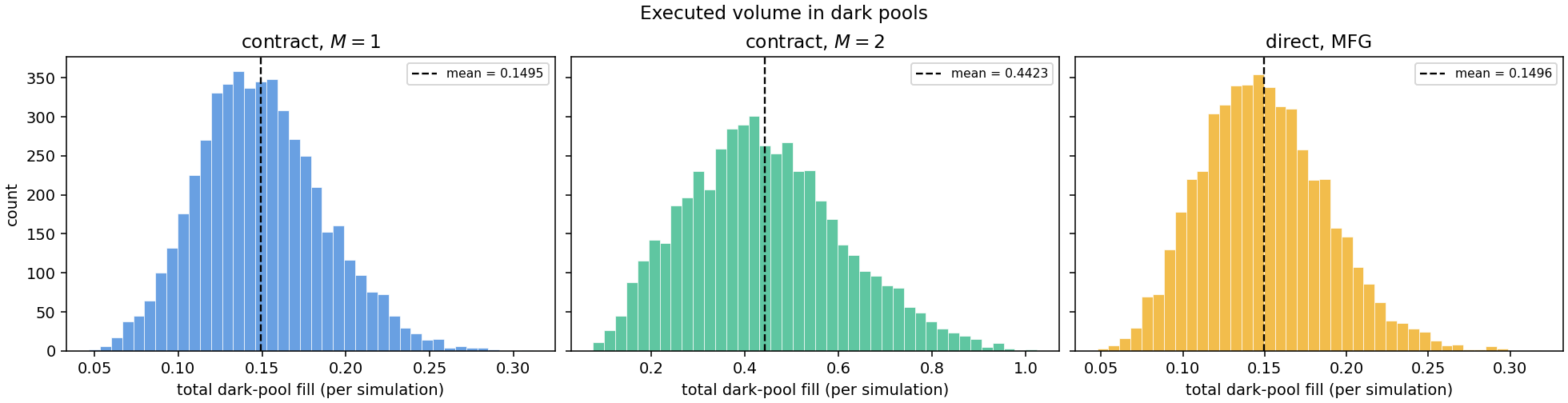}
    \caption{Total executed volume in dark pools per simulation, summed across the $M$ pools.}
    \label{fig:rev_dark_volume}
\end{figure}

\paragraph*{Permanent price impact.}
Figure~\ref{fig:rev_impact} reports the permanent price impact $S_T - S_0$ both as an absolute quantity (top row) and normalized by the realized client liquidation $1 - Q^C_T$ (bottom row). The per-unit metric isolates how much price each unit liquidated costs, independent of how much each regime chose to liquidate. Under direct execution, adding the block venue cuts the per-unit impact by nearly half (from $-0.022$ to $-0.013$), consistent with the analytical structure of Theorem~\ref{thm:mfg_hft_existence}: the optimal $\hat\nu^0$ feedback responds to a larger pool of substitutable dark liquidity by shifting flow off the lit venue. In the contract regime the same direction shows up qualitatively, but the broker also commits to a larger total volume at $M=2$ (near-full liquidation versus a $22\%$ residual at $M=1$), so the absolute impact rises while the per-unit metric remains comparable.

\begin{figure}[htbp]
    \centering
    \includegraphics[width=\textwidth]{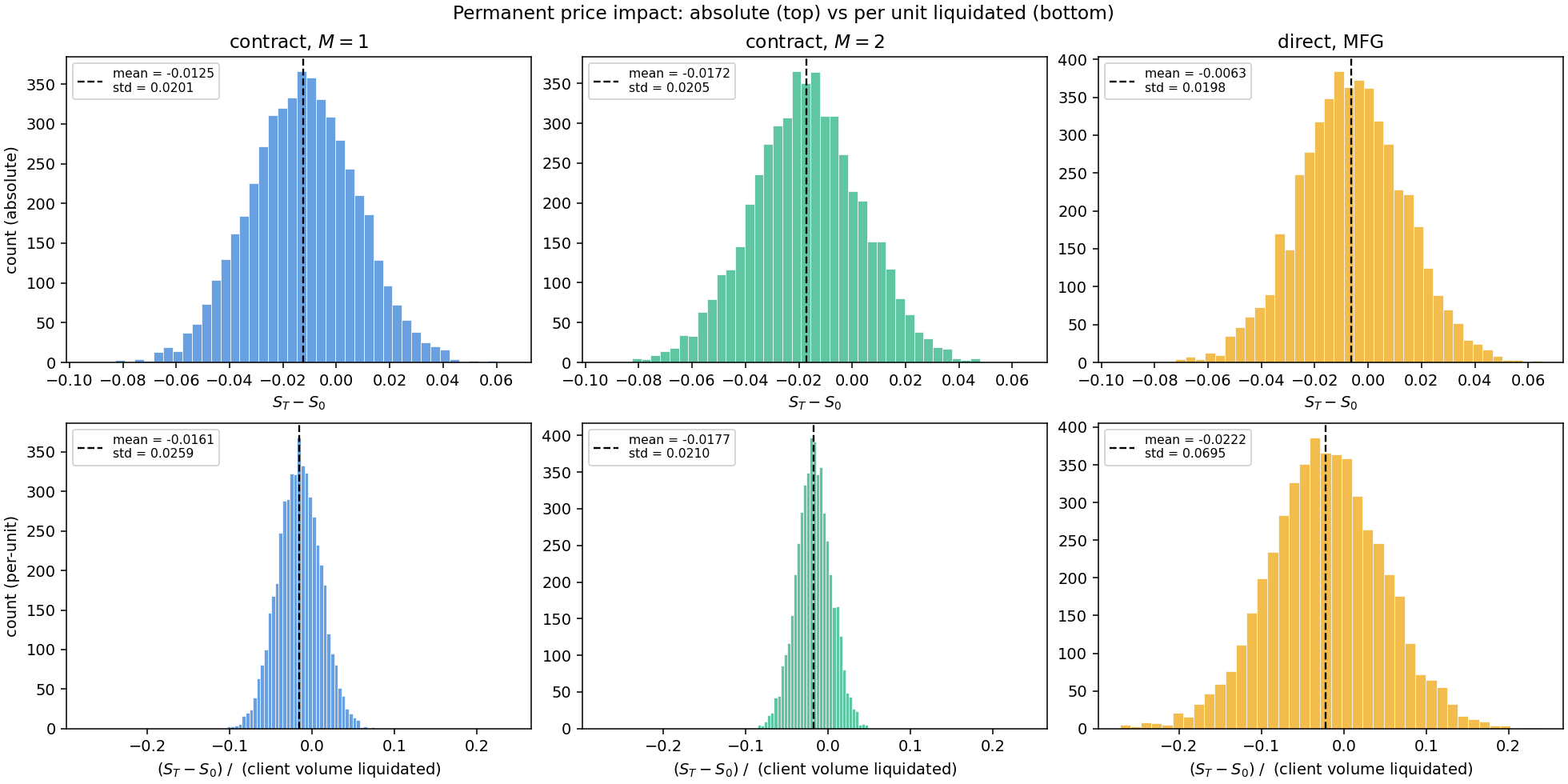}
    \caption{Permanent price impact distribution. Top row: absolute $S_T - S_0$. Bottom row: per-unit $(S_T - S_0)/(1 - Q^C_T)$, i.e., impact normalized by the client volume actually liquidated. The three regimes liquidate different totals at these parameters, so per-unit impact is the apples-to-apples metric.}
    \label{fig:rev_impact}
\end{figure}

\paragraph*{Compensation and the value of broker delegation.}
The economic split is summarized by Figures~\ref{fig:rev_compensation} and~\ref{fig:rev_pnl_split}. The mean compensation $\xi$ drops from $0.188$ at $M=1$ to $0.142$ at $M=2$, a $25\%$ reduction: with the block venue at its disposal the broker can offer a near-full-liquidation service at a discount, since the residual lit pressure that needs to be priced into $\xi$ is smaller. Figure~\ref{fig:rev_pnl_split} decomposes the contract's terminal PnL on both sides and adds the total surplus $\Pi^C_T + \Pi^B_T$ as the bottom row, with the $\xi$ transfer canceling on both sides:
\begin{center}
\begin{tabular}{lccc}
\hline
 & $M=1$ & $M=2$ & change \\
\hline
investor $\Pi^C_T$         & $0.809$ & $0.974$ & $+20\%$ \\
broker $\Pi^B_T$            & $0.151$ & $0.100$ & $-34\%$ \\
total $\Pi^C_T + \Pi^B_T$  & $0.961$ & $1.074$ & $+12\%$ \\
\hline
\end{tabular}
\end{center}
The most striking line is the last one: with the block venue available, the total surplus of the delegation exceeds the initial mark-to-market value $Q_0 S_0 = 1$ of the position. At $M=1$ the coalition loses about $4\%$ of the position's face value to price impact and residual-inventory risk; the block venue turns this into a $7\%$ joint gain, because a large fraction of the order exits at the mid-price through dark fills rather than walking the lit book. The split of this larger surplus, however, is entirely tilted toward the investor. This is not a violation of the broker's rationality: the participation constraint binds at the optimum (Remark~\ref{rem:linear_utility_inv}), so the broker's utility is pinned at the reservation level $R_0$ in both configurations and the move from $M=1$ to $M=2$ is a weak Pareto improvement in which the broker is exactly indifferent and the investor collects the whole gain. The drop in the broker's expected PnL (from $0.151$ to $0.100$) is the accounting counterpart of this indifference: the broker's mean profit is essentially the risk premium on a book whose utility value is fixed, and the premium shrinks as the block venue lets the broker unwind its internalized inventory more completely by the terminal date.

\begin{figure}[htbp]
    \centering
    \includegraphics[width=0.8\textwidth]{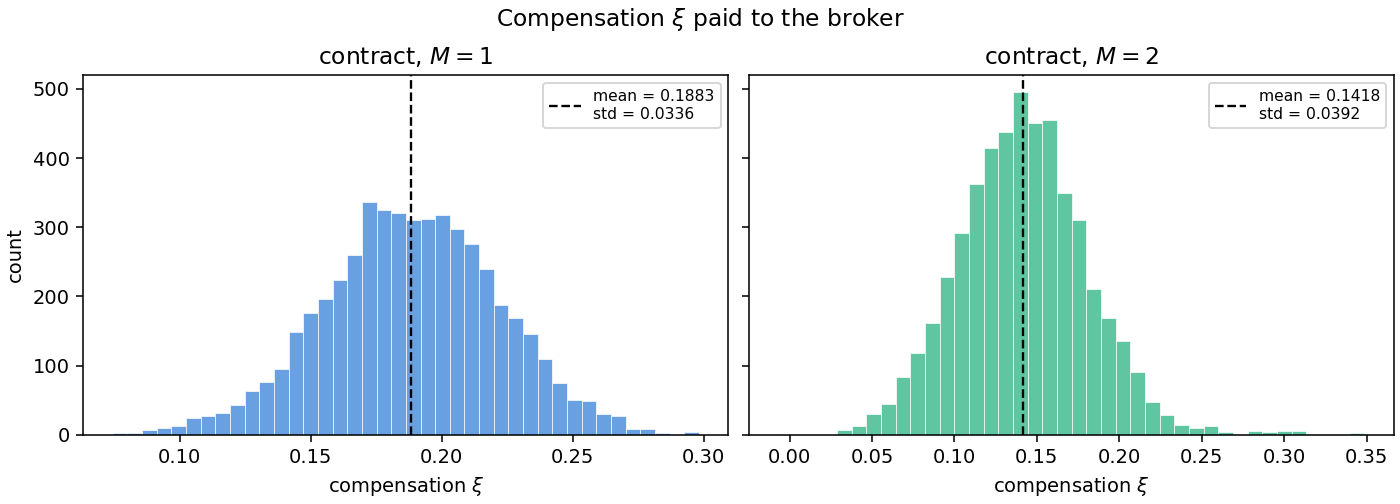}
    \caption{Distribution of the compensation $\xi$ paid by the investor to the broker.}
    \label{fig:rev_compensation}
\end{figure}

\begin{figure}[htbp]
    \centering
    \includegraphics[width=0.95\textwidth]{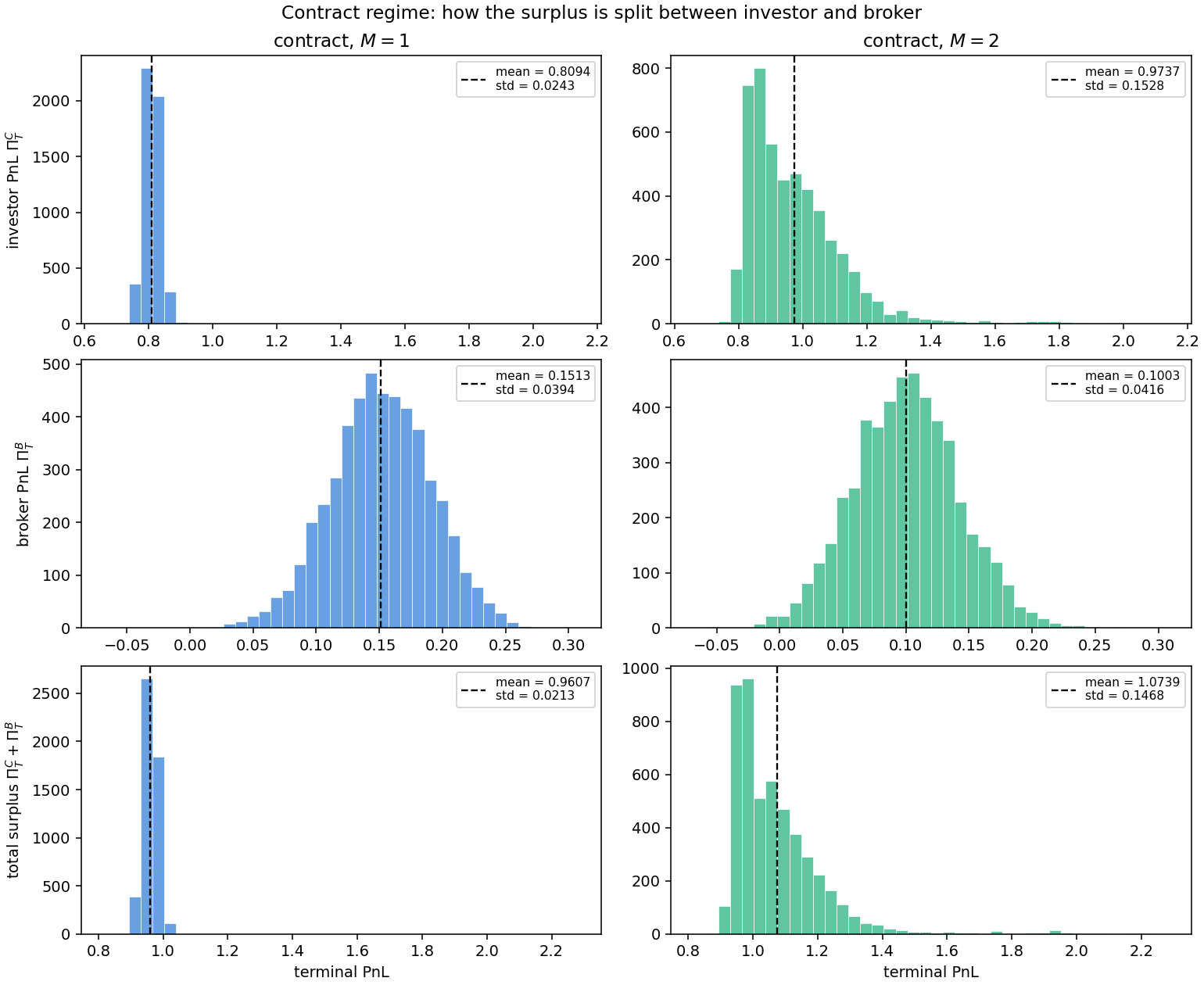}
    \caption{Terminal-PnL distributions in the contract regime. Top: investor $\Pi^C_T$. Middle: broker $\Pi^B_T$. Bottom: total surplus $\Pi^C_T + \Pi^B_T$, with the $\xi$ transfer canceling on both sides. Columns are $M=1$ (left) and $M=2$ (right); the $x$-axis range is shared row-wise for direct comparison.}
    \label{fig:rev_pnl_split}
\end{figure}
\vspace{-0.5em}
\paragraph*{Financial insights.}
The clearest analytical comparison is the direct-execution panel of Figure~\ref{fig:rev_impact} (right column): adding the block venue cuts the per-unit price impact by nearly half (from $-0.022$ to $-0.013$). The contract regime tells a complementary story through the broker-investor decomposition of Figures~\ref{fig:rev_broker} and~\ref{fig:rev_pnl_split}: access to the block venue lifts the total surplus of the delegation by $12\%$, above the face value of the block itself ($1.074 > Q_0 S_0 = 1$), while the compensation paid to the broker falls by $25\%$. In utility terms the move is a weak Pareto improvement: the broker's participation constraint binds in both configurations, so the broker is indifferent while the investor collects the entire gain ($+20\%$ in expected terminal PnL). In expected-PnL terms the broker's slice contracts to the risk premium of a book that the block venue lets it unwind almost completely. Finally, the allocation across the two heterogeneous venues shows that what matters for the broker is expected fill depth per match rather than match frequency: the rarely-matching block venue receives the larger share of the order.

\section{Conclusion}\label{sec:conclusion}
This study investigates the delegation of a large liquidation order to a broker operating in a hybrid market structure comprising one lit venue and multiple broker-owned dark pools. The delegation is modeled as a continuous-time principal-agent problem in which the broker enjoys an information privilege over the rest of the market, may internalize part of the client's flow against its own book, and hedges the resulting inventory on the lit venue. The broker's best response is characterized in closed form through a quadratic BSDE with jumps, and the investor's contracting problem reduces to a Markovian control problem solved by a verification theorem. A major-minor mean-field game, in which the investor liquidates directly against a crowd of strategic high-frequency traders, provides the no-broker benchmark and admits explicit solutions through a Riccati equation and a linear two-point boundary-value problem. Our numerical experiments indicate that the broker's bundled services command a compensation that decreases with the number of dark venues at the broker's disposal, and that access to an additional, heterogeneous dark pool both enlarges the total surplus of the delegation and redistributes it toward the investor. In the direct-execution benchmark, the same additional venue substantially reduces the per-unit price impact borne by the investor. These findings suggest that dark-pool fragmentation, often discussed from the angle of price discovery, also has first-order consequences for the division of surplus between institutional clients and their brokers.\vspace{0.5em}

For further studies, the present model assumes that the broker possesses perfect knowledge of the intensity of the Poisson processes governing dark-pool matching and of the distribution of liquidity across venues. A significant extension involves relaxing this assumption to reflect the realistic scenario where the broker must estimate the latent liquidity arrival rate and posterior liquidity distribution from observable market data. This necessitates integrating forward inference with backward optimal control within a Forward-Backward Stochastic Differential Equation (FBSDE) framework, and would rigorously model the broker's simultaneous learning and strategic adaptation under uncertainty. Another natural direction is competition between several brokers for the same client, or a single broker serving several clients with potentially conflicting execution schedules, which would turn the present bilateral contract into a richer market-design problem.

\newpage
\bibliographystyle{alpha}
\bibliography{dark_pools}

\end{document}